\pdfoutput=1
\documentclass[11pt]{article}

\usepackage[preprint]{acl}

\usepackage[utf8]{inputenc}
\usepackage[T1]{fontenc}
\usepackage[amsthm]{newtx}  % all fonts
\usepackage{inconsolata}  % monospace
\usepackage{microtype}  % adjust hyphenation
\usepackage[linesnumbered,ruled,vlined]{algorithm2e}
\usepackage{xpatch}
\usepackage{adjustbox}  % centering overfull table
\usepackage{arydshln}
\usepackage{booktabs}
\usepackage{multirow}
\usepackage{stfloats}
\usepackage[labelformat=parens]{subcaption}
\usepackage{pifont}
\usepackage{enumitem}
\usepackage[capitalise,noabbrev]{cleveref}  % [after hyperref]
\usepackage{xurl}  % proper line break for url

\usepackage{textcomp}
\usepackage{gensymb}

\xpatchcmd{\algocf@Vline}{\vrule}{\vrule \kern-0.4pt}{}{} % make algorithm number correctly aligned
\xpatchcmd{\@algocf@start}{-1.5em}{0pt}{}{}  % remove right hand margin in algorithm
\newcommand{\cmark}{\ding{51}\xspace}
\newcommand{\xmark}{\ding{55}\xspace}
\theoremstyle{definition}
\newtheorem{definition}{Definition}
\newtheorem{lemma}{Lemma}
\newtheorem{theorem}{Theorem}

\newcommand{\ours}{RemoteRAG\xspace}

\title{\ours: A Privacy-Preserving LLM Cloud RAG Service}

\author{
Yihang Cheng\textsuperscript{1}, Lan Zhang\textsuperscript{1}, Junyang Wang\textsuperscript{1}, Mu Yuan\textsuperscript{2}, Yunhao Yao\textsuperscript{1} \\
\textsuperscript{1}University of Science and Technology of China \\
\textsuperscript{2}The Chinese University of Hong Kong \\
\small\texttt{yihangcheng@mail.ustc.edu.cn, zhanglan@ustc.edu.cn, iswangjy@mail.ustc.edu.cn} \\
\small\texttt{muyuan@cuhk.edu.hk, yaoyunhao@mail.ustc.edu.cn}
}

\begin{document}
\maketitle

\begin{abstract}
    Retrieval-augmented generation (RAG) improves the service quality of large language models by retrieving relevant documents from credible literature and integrating them into the context of the user query.
    Recently, the rise of the cloud RAG service has made it possible for users to query relevant documents conveniently.
    However, directly sending queries to the cloud brings potential privacy leakage.
    In this paper, we are the first to formally define the privacy-preserving cloud RAG service to protect the user query and propose \ours as a solution regarding privacy, efficiency, and accuracy.
    For privacy, we introduce $(n,\epsilon)$-DistanceDP to characterize privacy leakage of the user query and the leakage inferred from relevant documents.
    For efficiency, we limit the search range from the total documents to a small number of selected documents related to a perturbed embedding generated from $(n,\epsilon)$-DistanceDP, so that computation and communication costs required for privacy protection significantly decrease.
    For accuracy, we ensure that the small range includes target documents related to the user query with detailed theoretical analysis.
    Experimental results also demonstrate that \ours can resist existing embedding inversion attack methods while achieving no loss in retrieval under various settings.
    Moreover, \ours is efficient, incurring only $0.67$ seconds and $46.66$KB of data transmission ($2.72$ hours and $1.43$ GB with the non-optimized privacy-preserving scheme) when retrieving from a total of $10^6$ documents.
\end{abstract}

\section{Introduction}

Large language models (LLMs) have attracted widespread attention since the release of ChatGPT \citep{ChatGPT}.
However, LLM is not without its flaws.
One major issue is its tendency to generate factually incorrect or purely fictional responses, a phenomenon known as hallucination \citep{DBLP_conf_chi_LeiserELKMSS24,DBLP_journals_corr_abs_2310_01469}.

To mitigate this problem, retrieval-augmented generation (RAG) \citep{DBLP_conf_nips_LewisPPPKGKLYR020} has been proposed to offer credible external knowledge, providing significant convenience for numerous tasks.
RAG aims to understand the input query, extract relevant information from external data sources, and enhance the quality of the generated answers \citep{DBLP_conf_icml_BorgeaudMHCRM0L22,DBLP_conf_nips_LewisPPPKGKLYR020,DBLP_journals_corr_abs_2306_06615}.
Specifically, RAG allows the retrieval of relevant documents which can help understand or answer the input query and inserts them into the context of prompts to improve the output of LLM \citep{DBLP_conf_eacl_IzacardG21}.
Its ability to enable LLM to provide answers with credible literature makes RAG an important technique in the application of LLM, leading to the development of many excellent and user-friendly open-source RAG projects \citep{githubGitHubLanggeniusdify,githubGitHubQuivrHQquivr,githubGitHubChatchatspaceLangchainChatchat}.

To leverage the power of RAG, a new concept RAG-as-a-Service (RaaS) has been proposed, gathering significant attention \citep{geniuseeServiceRetrieval,Nuclia}. 
In RaaS, the RAG service is entirely hosted online in the cloud.
The user submits requests to the cloud with input queries to receive responses from RaaS.
In this scenario, the cloud serves as the maintainer of the RAG service.
While the current solution facilitates the wide adoption of RaaS, it raises serious privacy concerns.
The input query may contain sensitive information, such as health conditions and financial status.
Unfortunately, this data is not protected and must be uploaded in plaintext to the cloud in order to retrieve relevant documents.
In this study, we aim to tackle a challenging question:
\textit{How can we minimize privacy leakage in queries for RaaS while ensuring the accuracy of the responses, all with minimal additional costs?}

Targeting the question above, we design a novel solution \ours.
\underline{For privacy}, we propose $(n,\epsilon)$-DistanceDP inspired by differential privacy and an embedding perturbation mechanism, so that the user can control the privacy leakage with a privacy budget $\epsilon$ in $n$-dimensional space of embeddings.
We further study the potential privacy leakage in averaging the most relevant embeddings and find it within the constraint of $(n,\epsilon)$-DistanceDP most of the time.
\underline{For efficiency}, we limit the search range from the total documents to a small number of selected documents, which are the relevant documents to a perturbed embedding generated from $(n,\epsilon)$-DistanceDP.
This small search range can save a significant amount of computation and communication costs used for privacy protection.
\underline{For accuracy}, we theoretically analyze the minimum size of the relevant documents to the perturbed embedding to ensure that they do contain target documents for the original query.

\textbf{Contributions} of this paper are listed as follows:
\begin{itemize}[leftmargin=*,noitemsep,topsep=0pt,wide=0pt,labelwidth=5pt]
    \item To the best of our knowledge, we are the first to address the privacy-preserving cloud RAG service problem.
          We formally define the privacy-preserving cloud RAG service and characterize its corresponding threat model.
    \item We propose \ours as a solution to the privacy-preserving cloud RAG service regarding privacy, efficiency, and accuracy.
          We define $(n,\epsilon)$-DistanceDP to characterize privacy leakage of the user query and design a mechanism to generate a perturbed embedding for the cloud for privacy, as well as to retrieve relevant documents within a minimum search range for efficiency.
          Accuracy is ensured by theoretical analysis of the minimum range produced by the perturbed embedding.
    \item We conduct extensive experiments to demonstrate that \ours can resist existing embedding inversion attack methods while achieving no loss in retrieval under various settings.
          The experiment results also show the efficiency of \ours, incurring only $0.67$ seconds and $46.66$KB of data transmission ($2.72$ hours and $1.43$ GB with the non-optimized privacy-preserving scheme) when retrieving from a total of $10^6$ documents.
\end{itemize}

\section{Problem Formulation}
\label{sec:formulation}

We first formally define the problem in developing a privacy-preserving LLM cloud RAG service.
The main notations used in this paper are listed in \cref{tab:notations} for ease of reference.

\begin{table}[t]
    \caption{Notation table.}
    \label{tab:notations}
    \centering
    \begin{tabular}{@{}ll@{}}
        \toprule
        Sym.       & Description                                 \\
        \midrule
        $N$        & Number of RAG documents in the cloud        \\
        $\epsilon$ & Privacy budget                              \\
        $e_k$      & User query embedding                        \\
        $k$        & Number of top documents related to $e_k$    \\
        $e_{k'}$   & Perturbed embedding                         \\
        $k'$       & Number of top documents related to $e_{k'}$ \\
        $n$        & Dimensional space of embeddings             \\
        \bottomrule
    \end{tabular}
\end{table}

\subsection{Problem Setup and Threat Model}
\label{subsec:setup}

One RAG request process involves two sides: a cloud and a user.
The cloud hosts a substantial number ($N$) of documents as a RAG service.
The user submits a request with a user query, and the RAG service in the cloud should retrieve top $k$ relevant documents.
An embedding model, shared between two sides, enables the user to convert the query into an embedding $e_k$.
Additionally, the user has a privacy budget $\epsilon$ intended to measure and limit the privacy leakage of the query.

\textbf{Threat model.}
We consider this scenario semi-honest, where both sides adhere to the protocol but the cloud is curious about the private information of the user query.
During the RAG request process, the user should not reveal the semantic information of the query beyond the privacy budget $\epsilon$ allows.
Apart from the query itself, we should also consider the protection of the query embedding and the indices of top $k$ documents:
\begin{itemize}[label=$\triangleright$,leftmargin=*,noitemsep,topsep=0pt,wide=0pt]
    \item \textit{Protection of the query embedding.}
          Existing attack methods \citep{DBLP_conf_emnlp_MorrisKSR23,DBLP_journals_corr_abs_2402_12784} have demonstrated that the semantic information can be extracted from the embedding if the embedding model is accessible.
          Consequently, safeguarding the semantic information of the query necessitates the protection of its embedding as well.
    \item \textit{Protection of the indices of top $k$ documents.}
          The embeddings of top $k$ documents are situated in proximity to the query embedding, which potentially leads to the leakage of the query embedding.
          Specifically, the average of top $k$ document embeddings could be close to the query embedding, for which the privacy leakage should be carefully studied.
\end{itemize}

\subsection{Design Scope}

In \ours, we examine the potential privacy leakage occurring during the data transmission between the cloud and the user.
Considerations of offline RAG services downloaded directly from the cloud or internal privacy issues \citep{DBLP_journals_corr_abs_2402_16893} specific to RAG are beyond the scope of this paper.

\begin{figure*}[t]
    \centering
    \includegraphics[width=\textwidth]{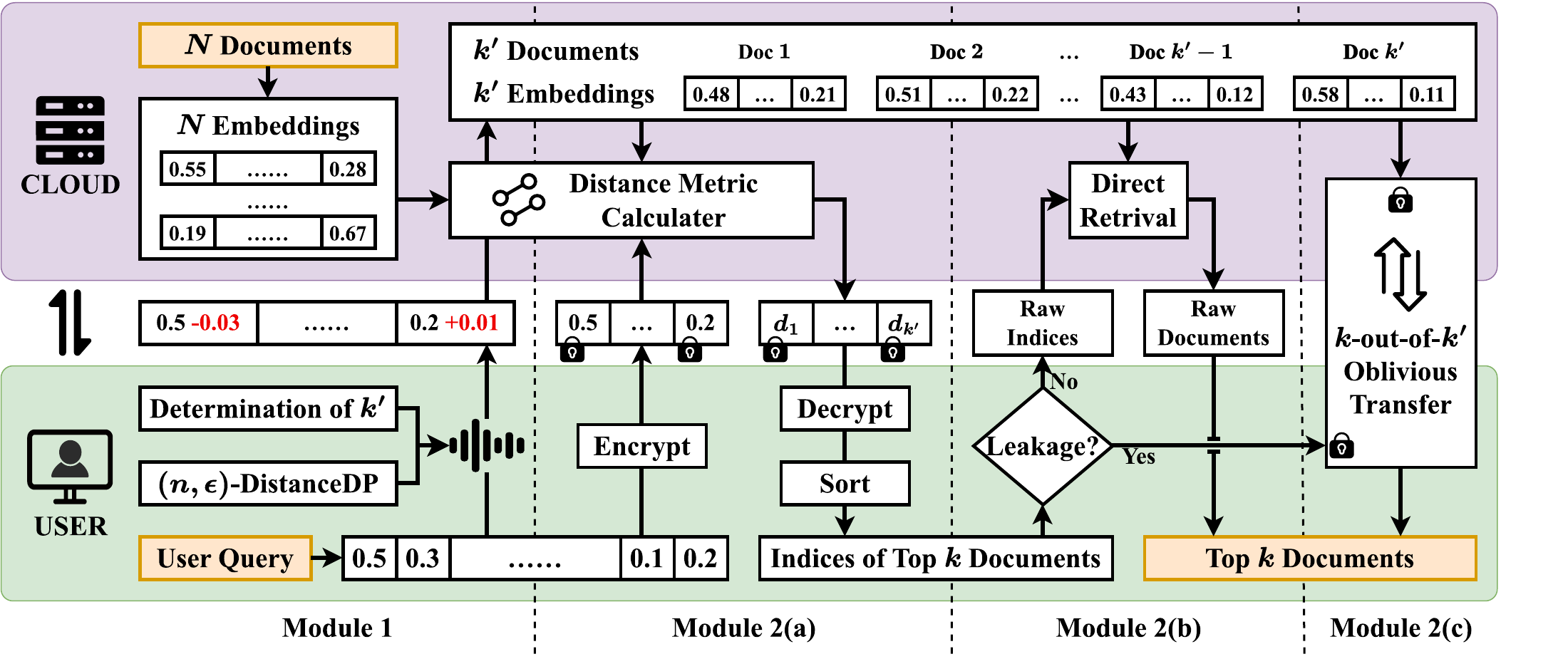}
    \caption{The flowchart of \ours. Module 1 preserves privacy with $(n,\epsilon)$-DistanceDP and improves efficiency by limiting the search range. Module 2 retrieves documents with different choices based on leakage circumstances.}
    \label{fig:overview}
\end{figure*}

\section{System Design}

\subsection{Overview}

The flowchart of \ours is shown in \cref{fig:overview}.
Under the control of $(n,\epsilon)$-DistanceDP, module \hyperref[subsec:module1]{1} aims to reduce the search range to enhance efficiency while ensuring accuracy.
Module \hyperref[subsec:module2]{2} targets safely retrieving top $k$ relevant documents within the limited range from two optional choices under different leakage circumstances.

\subsection{Range Limitation with Privacy Budget}
\label{subsec:module1}

\subsubsection{Generation of Perturbation}
\label{subsubsec:module1_generation}

Given the requirement that the query embedding should not be transmitted to the cloud, we opt to send a perturbed embedding instead.
To assess the potential privacy leakage, we utilize the differential privacy (DP) theory to define $(n,\epsilon)$-DistanceDP in $n$-dimensional space:

\begin{definition}[$(n,\epsilon)$-DistanceDP]
    A mechanism $K$ satisfies $(n,\epsilon)$-DistanceDP if and only if $\forall x,x'\in\mathbb{R}^n$:
    \[L(K(x),K(x'))\le\epsilon\|x-x'\|\]
    where $\epsilon$ is a given privacy budget, $\|x-x'\|$ denotes L2 distance, and $L(K(x),K(x'))=\ln\frac{Pr(K(x)=y)}{Pr(K(x')=y)}$ represents the distance between the probabilities of any target value $y$ drawn from the distributions $K(x)$ and $K(x')$ generated by points $x$ and $x'$.
\end{definition}

To apply $(n,\epsilon)$-DistanceDP in \ours, for any query embedding $e_k$, the user generates a perturbed embedding $e_{k'}$ using a noise function, which should satisfy that from the perspective of $e_{k'}$, the probability of generating the query embedding $e_k$ and another random embedding $e_x$ around $e_{k'}$ with the noise function differs by at most a multiplicative factor of $e^{-\epsilon\|e_k-e_x\|}$.

The property above can be achieved by utilizing the Laplace distribution, as discussed in \citep{DBLP_conf_ccs_AndresBCP13,DBLP_conf_tcc_DworkMNS06}.
The primary difference lies in our application within higher $n$-dimensional space.
Given the privacy budget $\epsilon\in\mathbb{R}^+$ and the actual point $x_0\in\mathbb{R}^n$, the probability density function (pdf) of the noise function at any other point $x\in\mathbb{R}^n$ is given by:
\[D_{n,\epsilon}(x|x_0)\propto e^{-\epsilon\|x-x_0\|}\]

Directly generating a point according to the above distribution is challenging.
We first focus on the radial component $r=\|x-x_0\|$, whose marginal distribution is given by:
\[D_{n,\epsilon}(r)\propto r^{n-1}e^{-\epsilon r}\]
which corresponds exactly to the pdf of the gamma distribution with shape parameter $n$ and scale parameter $\frac{1}{\epsilon}$.
Therefore, the radial component $r\sim D_{n,\epsilon}(r)=\mathrm{Gamma}(n,\frac{1}{\epsilon})$.

Next, the direction vector is sampled from a uniform distribution on the $n$-dimensional unit sphere.
This is accomplished by independently sampling $t_i\sim N(0,1)$ from the standard normal distribution and then normalizing $t_i$ by $\frac{t_i}{\sqrt{\sum_{j=1}^{n}{t_j}^2}},i\in[1,n]$.

\begin{figure}[t]
    \centering
    \includegraphics[width=0.93\columnwidth]{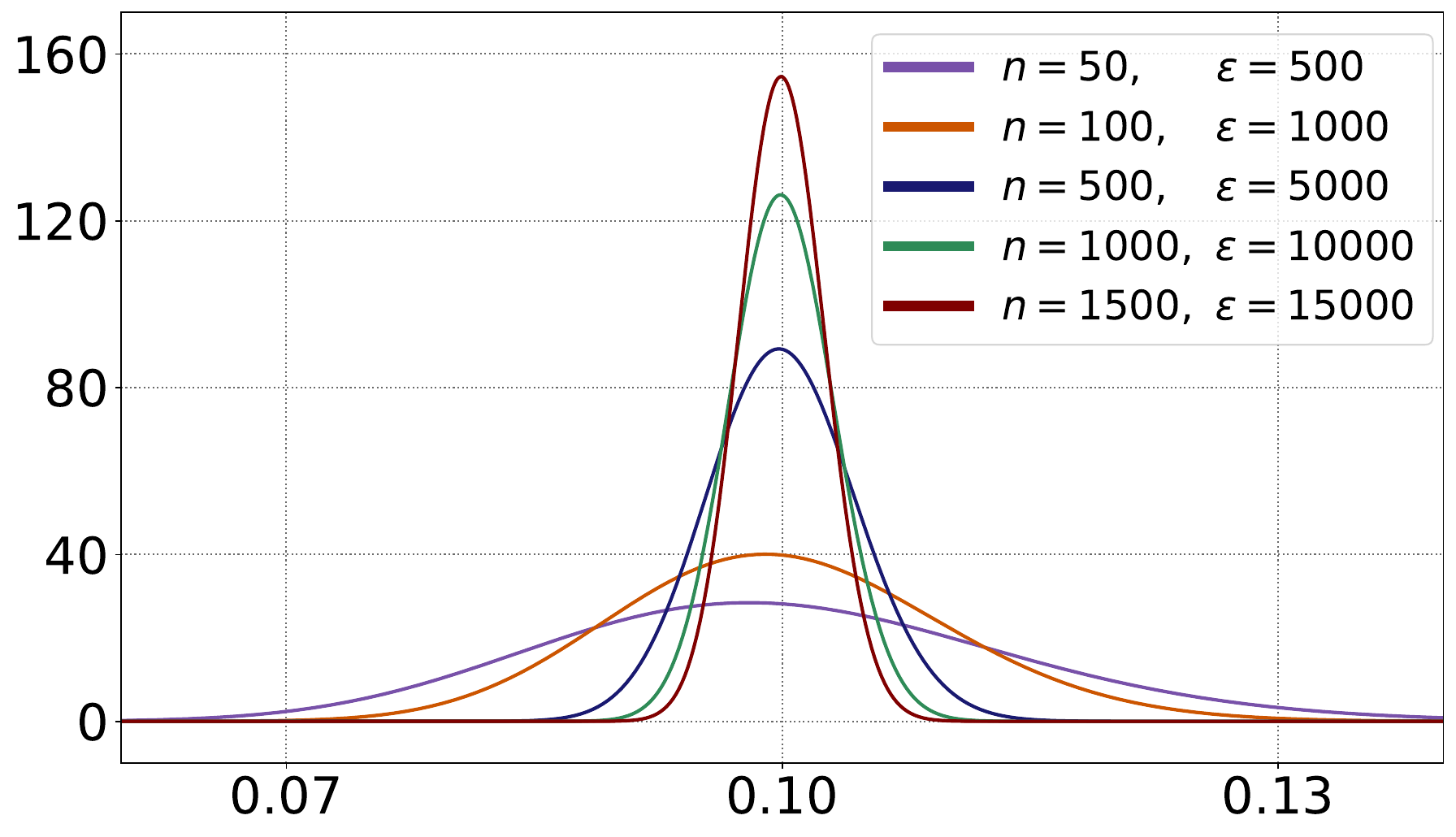}
    \caption{The probability density function of different gamma distributions within $[0.06,0.14]$ range. ($\epsilon=10n$)}
    \label{fig:gamma}
\end{figure}

\textbf{\boldmath Appropriate range for $\epsilon$.}
We offer the following guidelines to help choose the privacy budget $\epsilon$.
Based on the previous analysis, we observe that the perturbation $r\sim\mathrm{Gamma}(n,\frac{1}{\epsilon})$, and the expected value of it is $\bar{r}=\frac{n}{\epsilon}$.
Current embedding models often produce embeddings with a large dimension (e.g., $384$ for all-MiniLM-L12-v2 \citep{huggingfaceSentencetransformersallMiniLML12v2Hugging}, $768$ for gtr-t5-base \citep{huggingfaceSentencetransformersgtrt5baseHugging}, and $1536$ for text-embedding-ada-002 \citep{OpenAI1}).
The pdf of $\mathrm{Gamma}(n,\frac{1}{\epsilon})$ becomes increasingly steep as the dimension $n$ increases, as illustrated in \cref{fig:gamma}.
This implies that the radial components drawn from the distribution are likely to cluster around $\bar{r}$.
Given that embeddings are already normalized, a typical perturbation should fall within the range of $0.02$ to $0.1$, as discussed in \cref{fig:attack-r}, which corresponds to the privacy budget $\epsilon$ ranging from $10n$ to $50n$.

\subsubsection{Calculation of Search Range}

After generating the perturbation, the user then requests the cloud to retrieve top $k'$ documents related to the perturbed embedding $e_{k'}$, thereby limiting the search range.
To maintain accuracy, it is crucial to ensure that these $k'$ documents include top $k$ documents related to the query embedding $e_k$.
This requirement motivates the need to determine the appropriate value for $k'$.

\begin{lemma}
    \label{lem:k_alpha_k}
    Assume that there are $N$ embeddings uniformly distributed on the surface of the $n$-dimensional unit sphere.
    Let $\alpha_k$ be the polar angle of the surface area formed by top $k$ embeddings related to any given embedding.
    Then, $k$ and $\alpha_k$ satisfy the following relationship:
    \[k=N\cdot\frac{\Omega_{n-1}(\pi)}{\Omega_{n}(\pi)}\cdot\int_{0}^{\alpha_{k}}\sin^{n-2}\theta\,\mathrm{d}\theta\]
    where $\Omega_{n}(\pi)=\frac{2\pi^\frac{n}{2}}{\Gamma(\frac{n}{2})}$ represents the surface area of the unit $n$-sphere.
\end{lemma}

\begin{figure}[t]
    \centering
    \begin{subcaptionblock}[t]{0.48\columnwidth}
        \centering
        \includegraphics[width=0.95\columnwidth]{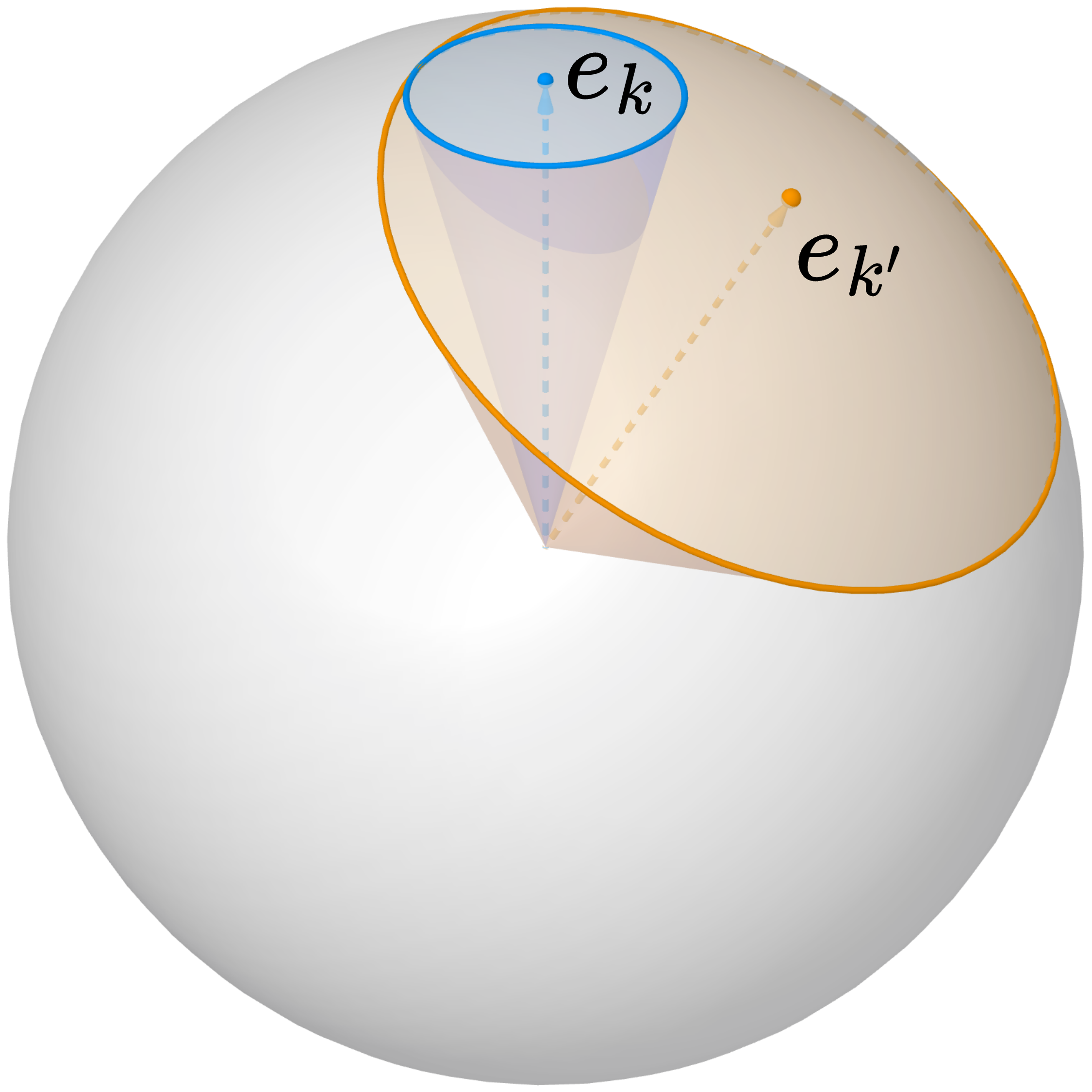}
        \caption{Oblique projection.\\ Top $k'$ documents related to $e_{k'}$ include top $k$ documents related to $e_k$.}
        \label{fig:thm1proof_global}
    \end{subcaptionblock}%
    \hfill%
    \begin{subcaptionblock}[t]{0.48\columnwidth}
        \centering
        \includegraphics[width=0.95\columnwidth]{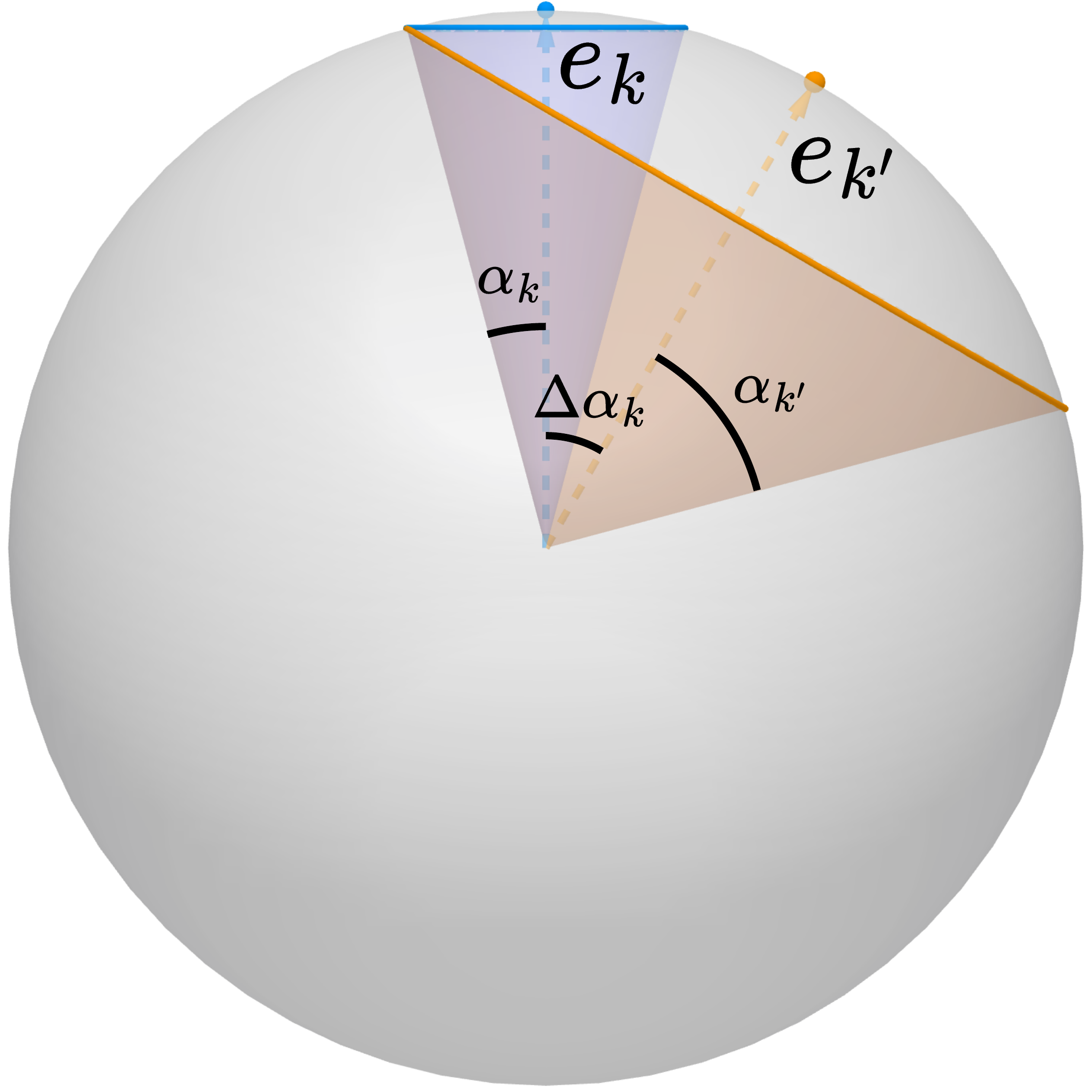}
        \caption{Orthographic projection.\\ $\alpha_{k'}=\alpha_k+\Delta\alpha_k$.}
        \label{fig:thm1proof_special}
    \end{subcaptionblock}
    \caption{Illustration in 3D projection.}
    \label{fig:thm1proof}
\end{figure}

From \cref{lem:k_alpha_k}, we can derive the polar angle $\alpha_k$ from $k$.
Since the perturbation is small ($r\ll1$), the perturbed angle $\Delta\alpha_k$ between $e_k$ and $e_{k'}$ can be approximately by $r$.
To ensure that top $k'$ documents related to the perturbed embedding $e_{k'}$ include top $k$ documents related to the query embedding $e_k$, we further propose \cref{thm:kpr_k} and illustrate the principle in \cref{fig:thm1proof}.

\begin{theorem}
    \label{thm:kpr_k}
    Under the conditions specified in \cref{lem:k_alpha_k},
    given two embeddings $e_k$ and $e_{k'}$ with the perturbed angle $\Delta\alpha_k$, to ensure that top $k'$ embeddings related to $e_{k'}$ include top $k$ embeddings related to $e_k$, $k'$ and $k$ satisfy the following relationship:
    \[\Delta k=k'-k=N\cdot\frac{\Omega_{n-1}(\pi)}{\Omega_{n}(\pi)}\cdot\int_{\alpha_k}^{\alpha_{k'}}\sin^{n-2}\theta\,\mathrm{d}\theta\]
    where $\alpha_{k'}=\alpha_{k}+\Delta\alpha_k$.
\end{theorem}

\subsection{Retrieval with Cryptographic Protection}
\label{subsec:module2}

\subsubsection{Homomorphic Encryption for Indices}
\label{subsubsec:module2_a}

After limiting the search range to $k'$ documents, we now focus on obtaining the indices of top $k$ documents for the query embedding $e_k$.
Based on documentations of several open-source vector databases such as ChromaDB \citep{githubGitHubChromacorechroma}, FAISS \citep{DBLP_journals_corr_abs_2401_08281}, and Elasticsearch \citep{githubGitHubElasticelasticsearch}, we find that only two distance metrics are being used by all of them and set to be default: L2 distance and cosine distance.

\begin{definition}[Distance metrics]
    \label{def:cos_dis}
    Given two normalized embeddings $e_a$ and $e_b$ of the same dimension, L2 distance and cosine distance are calculated as follows:
    \begin{align*}
        d_{l2}(e_a,e_b)&=\|e_a-e_b\| \\
        d_{\cos}(e_a,e_b)&=1-\frac{\langle e_a,e_b\rangle}{\|e_a\|\cdot\|e_b\|}=1-\langle e_a,e_b\rangle
    \end{align*}
\end{definition}

Further analysis in \cref{thm:distance} reveals that using which distance metric does not affect the ranking of normalized vectors.
Therefore, we only consider cosine distance as the standard metric in this paper.

\begin{theorem}
    \label{thm:distance}
    Given two normalized embeddings $e_a$ and $e_b$ of the same dimension, L2 distance and cosine distance have the following relationship:
    \[d_{l2}(e_a,e_b)=\sqrt{2d_{\cos}(e_a,e_b)}\]
\end{theorem}

Recall that the query embedding $e_k$ should not be revealed to the cloud.
Considering that cosine distance involves only linear operations, we propose using partially homomorphic encryption (PHE).
Compared to fully homomorphic encryption, PHE is more computationally efficient and sufficient for calculating cosine distance.

Specifically, the user encrypts the query embedding $e_k$ using PHE and sends the encrypted form $\llbracket e_k\rrbracket$ to the cloud.
The cloud then calculates cosine distances $\llbracket d_{\cos}(e_k,e)\rrbracket=d_{\cos}(\llbracket e_k\rrbracket,e)$ in encrypted form between $\llbracket e_k\rrbracket$ and each document embedding $e$ related to $e_{k'}$.
Upon receiving these distances, the user decrypts them and sorts the results to find the indices of top $k$ documents related to the query.

\subsubsection{Safe Retrieval with Indices}
\label{subsubsec:module2_b}

We then carefully analyze whether these indices are safe to be directly sent to the cloud to retrieve the corresponding documents.

If the cloud receives these indices, it also means that it knows which $k$ documents are the closest to the user query.
The cloud can average these document embeddings to construct an embedding $\bar{e}$ which approximates the query embedding $e_k$.
Therefore, we should measure how close the two embeddings $e_k$ and $\bar{e}$ are.

\begin{theorem}
    \label{thm:index_privacy}
    Given a target query embedding $e_k$ and the mean embedding $\bar{e}$ of top $k$ relevant document embeddings, the mean angle $\omega$ between $e_k$ and $\bar{e}$ satisfies
    \begin{equation*}
        \tan\omega=\frac{\tan\alpha_k}{\sqrt{k}}
    \end{equation*}
    where $\alpha_k$ is calculated from \cref{lem:k_alpha_k}.
\end{theorem}

From \cref{thm:index_privacy}, we characterize the approximation between $e_k$ and $\bar{e}$ with the mean angle $\omega$.
Recall that the privacy budget $\epsilon$ generates a perturbation with a mean of $\frac{n}{\epsilon}$ (i.e.\ $\Delta\alpha_k\approx\bar{r}=\frac{n}{\epsilon}$), as described in \cref{subsubsec:module1_generation}.
Thus, if $\omega\ge\Delta\alpha_k$, $\bar{e}$ is within the control of the privacy budget $\epsilon$ and the indices require no extra protection.

\phantomsection
\label{subsubsec:module2_c}
\textbf{\boldmath Optional $k$-out-of-$k'$ oblivious transfer (OT) for safe retrieval.}
If $\omega<\Delta\alpha_k$, $\bar{e}$ is even closer to the query embedding $e_k$ than the perturbed embedding $e_{k'}$, therefore requiring further protection.
In this situation, we suggest using the $k$-out-of-$k'$ OT protocol \citep{DBLP_conf_latincrypt_ChouO15}, which allows a sender (the cloud) with a set of $k'$ messages (documents) to transfer a subset of $k$ messages to a receiver (the user) while remaining oblivious to the specific subset (indices) chosen by the receiver.
The OT protocol is described in \cref{subsec:ot}.

\begin{table*}[t]
      \caption{Comparison among \ours, privacy-ignorant and privacy-conscious services.}
      \label{tab:special_cases}
      \centering
      \begin{tabular}{@{}l@{ }lc@{}cccc@{}}
            \toprule
                                   &                         & Security                                   & & \multicolumn{3}{c}{Communication}                                         \\
                                                               \cmidrule(lr){3-4}                             \cmidrule(l){5-7}
                                   &                         & User Query                                 & & Rounds               & Numbers ($\beta$ units) & Documents ($\eta$ units) \\
            \midrule
            \multicolumn{2}{@{}l}{Privacy-ignorant Service}  & \textcolor[rgb]{0.7,0.15,0.15}{\xmark}     & & $1$                  & $n$                     & $k$                      \\
            \multicolumn{2}{@{}l}{Privacy-conscious Service} & \textcolor[rgb]{0,0.5,0}{\cmark}           & & $2$                  & $n+2N+1$                & $N$                      \\
            \\[-12pt]\hdashline\\[-12pt]
            \multirow{2}{*}{\ours} & (Direct)                & \multirow{2}{*}{$(n,\epsilon)$-DistanceDP} & & \multirow{2}{*}{$2$} & $2n+k+k'+1$             & $k$                      \\
                                   & (OT)                    &                                            & &                      & $2(n+k'+1)$             & $k'$                     \\
            \bottomrule
      \end{tabular}
\end{table*}

\section{Analysis on \ours}

\subsection{Security Analysis}

\ours must adhere to the privacy-preserving goal outlined in the threat model in \cref{subsec:setup}.
The cloud receives the following messages which may leak the semantic meaning of the user query:
\begin{enumerate}[label=(\arabic*),leftmargin=*,noitemsep,topsep=0pt,wide=0pt]
    \item[\underline{Module \hyperref[subsec:module1]{1}.}] The cloud receives the perturbed embedding.
          The perturbation $r$ between the user embedding and the perturbed embedding is sampled from $\mathrm{Gamma}(n,\frac{1}{\epsilon})$, which satisfies $(n,\epsilon)$-DistanceDP, controlled by the privacy budget $\epsilon$.
    \item[\underline{Module \hyperref[subsubsec:module2_a]{2(a)}.}] The cloud receives the encrypted form of the query embedding.
           Without the secret key, the cloud cannot reverse engineer the query embedding.
          The computation of cosine distances is guaranteed by PHE, which does not leak any information either.
    \item[\underline{Module \hyperref[subsubsec:module2_b]{2(b)}, if $\omega\ge\Delta\alpha_k$.}] The cloud receives the indices of top $k$ relevant documents.
          The perturbation from the mean embedding of top $k$ relevant document embeddings is within the protection scope of the privacy budget $\epsilon$.
    \item[\underline{Module \hyperref[subsubsec:module2_c]{2(c)}, if $\omega<\Delta\alpha_k$.}] The cloud and the user perform the $k$-out-of-$k'$ OT protocol.
          The property of the OT protocol ensures that the indices of top $k$ relevant documents are not visible to the cloud.
\end{enumerate}

From the analysis above, we demonstrate that the user receives top $k$ documents without disclosing any information about the user query under the constraint of the given privacy budget $\epsilon$.

\subsection{Communication Analysis}

We analyze communication from two aspects: the number of communication rounds and the size of communication.
We define one communication round as the transmission of a message from one side to another and back to the original side.
The size of one number and one document are set to be $\beta$ units and $\eta$ units, respectively.
\begin{enumerate}[label=(\arabic*),leftmargin=*,noitemsep,topsep=0pt,wide=0pt]
    \item[\underline{Module \hyperref[subsec:module1]{1}.}]
    There is one message transmitted from the user to the cloud ($0.5$ communication round), containing the perturbed embedding $e_{k'}$ and the corresponding $k'$.
    $e_{k'}$ is a vector of length $n$ with $n\beta$ units, while $k'$ is a number occupying $\beta$ units.
    \item[\underline{Module \hyperref[subsubsec:module2_a]{2(a)}.}]
    There is $1$ communication round.
    First, the encrypted form $\llbracket e_k\rrbracket$ of the query embedding is sent to the cloud, occupying $n\beta$ units.
    Second, the cloud sends back encrypted cosine distances, occupying $k'\beta$ units.
    \item[\underline{Module \hyperref[subsubsec:module2_b]{2(b)}, if $\omega\ge\Delta\alpha_k$.}]
    There is $1$ communication round.
    The user sends the indices ($k\beta$ units) to the cloud and the cloud returns the target documents ($k\eta$ units).
    \item[\underline{Module \hyperref[subsubsec:module2_c]{2(c)}, if $\omega<\Delta\alpha_k$.}]
    There are $1.5$ communication rounds and $(k'+1)\beta+k'\eta$ units of messages for the $k$-out-of-$k'$ OT protocol.
    Details can be found in \cref{subsec:ot}.
\end{enumerate}

By summing these, if $\omega\ge\Delta\alpha_k$, the total number of communication rounds is $2.5$, and the size of communication is $(2n+k+k'+1)\beta+k\eta$ units; if $\omega<\Delta\alpha_k$, the total number of communication rounds is $3$, and the size of communication is $2(n+k'+1)\beta+k'\eta$ units.

\textbf{Practical optimization.}
The number of communication rounds can be further reduced in practice.
For example, the user can simultaneously send both the perturbed embedding and the encrypted form of the query embedding to the cloud in a single communication to reduce $0.5$ round for modules \hyperref[subsec:module1]{1} and \hyperref[subsubsec:module2_a]{2(a)}.
Additionally, the cloud can send the encrypted cosine distances and start the OT protocol together to reduce another $0.5$ round for modules \hyperref[subsubsec:module2_a]{2(a)} and \hyperref[subsubsec:module2_c]{2(c)}.
Therefore, no matter whether with module \hyperref[subsubsec:module2_b]{2(b)} or \hyperref[subsubsec:module2_c]{2(c)}, the total number of communication rounds can be further reduced to $2$.

\subsection{Special Cases}
\label{subsec:special}

There are two special cases by varying the privacy budget $\epsilon$, which serve as two baselines in our experiments:

\textbf{The privacy-ignorant cloud RAG service.}
A privacy-ignorant cloud RAG service does not account for user query privacy, requiring the user to upload the query embedding and receive top $k$ documents directly.
This represents a special case of \ours, achieved by setting $\epsilon\to\infty$, with the perturbation $r\sim\mathrm{Gamma}(n,0)$ (i.e. no perturbation).
The service requires $1$ communication round with $n\beta+k\eta$ units in this case.

\textbf{The privacy-conscious cloud RAG service.}
A privacy-conscious cloud RAG service aims to fully protect user query privacy.
This can be regarded as the combination of modules \hyperref[subsubsec:module2_b]{2(a)} and \hyperref[subsubsec:module2_b]{2(c)} in \ours, where $k'=N$.
This is another special case of \ours, achieved by setting $\epsilon\to0$, with the perturbation $r\sim\mathrm{Gamma}(n,\infty)$ and $k'=N$ (i.e. cryptographic computation over all $N$ documents).
This case requires $2$ communication rounds with $(n+2N+1)\beta+N\eta$ units.

The comparison between \ours and these two special cases is presented in \cref{tab:special_cases}.

\section{Experiments}

We evaluate \ours under various settings.
Our key findings are as follows:
\begin{itemize}[label=$\triangleright$,leftmargin=*,noitemsep,topsep=0pt,wide=0pt]
    \item For privacy, \ours controls the semantic information leakage of the user query with the privacy budget.
          \textbf{[\cref{subsec:privacy}]}
    \item For accuracy, \ours achieves lossless document retrieval.
          \textbf{[\cref{subsec:accuracy}]}
    \item For efficiency, \ours introduces little extra computation and communication costs while preserving privacy.
          \textbf{[\cref{subsec:efficiency}]}
\end{itemize}

\subsection{Experimental Setup}

\textbf{Datasets and vector database.}
For the RAG dataset, we employ MS MARCO \citep{DBLP_conf_nips_NguyenRSGTMD16}.
Document embeddings are stored in Chroma \citep{githubGitHubChromacorechroma}, serving as the vector database.

\textbf{Embedding models.}
We use a total of five embedding models to measure the impact of different dimensions, including three open-sourced embedding models all-MiniLM-L12-v2 (MiniLM) \citep{huggingfaceSentencetransformersallMiniLML12v2Hugging}, all-mpnet-base-v2 (MPNet) \citep{huggingfaceSentencetransformersallmpnetbasev2Hugging}, gtr-t5-base (T5) \citep{huggingfaceSentencetransformersgtrt5baseHugging}, and two OpenAI embedding models text-embedding-ada-002 (OpenAI-1) \citep{OpenAI1}, text-embedding-3-large (OpenAI-2) \citep{OpenAI2}.

\textbf{Baselines.}
Given that current RAG research does not yet address the privacy issue, we compare \ours with the privacy-ignorant and privacy-conscious services as described in \cref{subsec:special}.

\subsection{Privacy Study}
\label{subsec:privacy}

We first examine privacy leakage and control with the privacy budget in \ours.
We apply the attack method Vec2Text \citep{DBLP_conf_emnlp_MorrisKSR23} and use SacreBLEU \citep{DBLP_conf_wmt_Post18} to measure the difference between the original query and the sentence reconstructed from the query embedding.
The embedding model is T5.

From an intuitive perspective, we plot the SacreBLEU metric achieved by Vec2Text against the perturbation $r$ to see how the perturbation affects the attack.
From the results in \cref{fig:attack-r}, we observe that the attack performance drops from $50$ to $10$ as the perturbation increases from $0$ to $0.2$.
When the perturbation reaches $0.2$, which is relatively large, the attack becomes completely ineffective.
However, this also results in a very large $k'$, leading to unacceptable computation and communication costs, as described later in \cref{fig:efficiency}.

From another perspective, we analyze the variation in the attack performance against the privacy budget $\epsilon$ to help choose a proper privacy budget.
The results are shown in \cref{fig:attack-eps}.
Overall, the performance of the attack improves as $\epsilon$ increases.
This is within our expectation, since a larger privacy budget means a looser tolerance for privacy leakage, which allows for a smaller perturbation and ultimately leads to a better attack performance back in \cref{fig:attack-r}.
Additionally, we find that the trend change is not linear: the change is sensitive when the value of $\epsilon$ is smaller than $50000$ but the attack performance is almost the same when the value of $\epsilon$ is larger than $100000$.
This is due to the relationship between $r$ and $\epsilon$:
Since $r\sim\mathrm{Gamma}(n,\frac{1}{\epsilon})$, it means $\bar{r}=\frac{n}{\epsilon}$, where $n=768$ (T5).
The perturbation quickly drops from $0.2$ to $0.016$ as $\epsilon$ varies from $3000$ to $50000$, and the attack performance improves from $10$ to $40$ in this range.

\begin{figure}[t]
    \centering
    \begin{subcaptionblock}[t]{0.5\columnwidth}
        \centering
        \includegraphics[width=\textwidth]{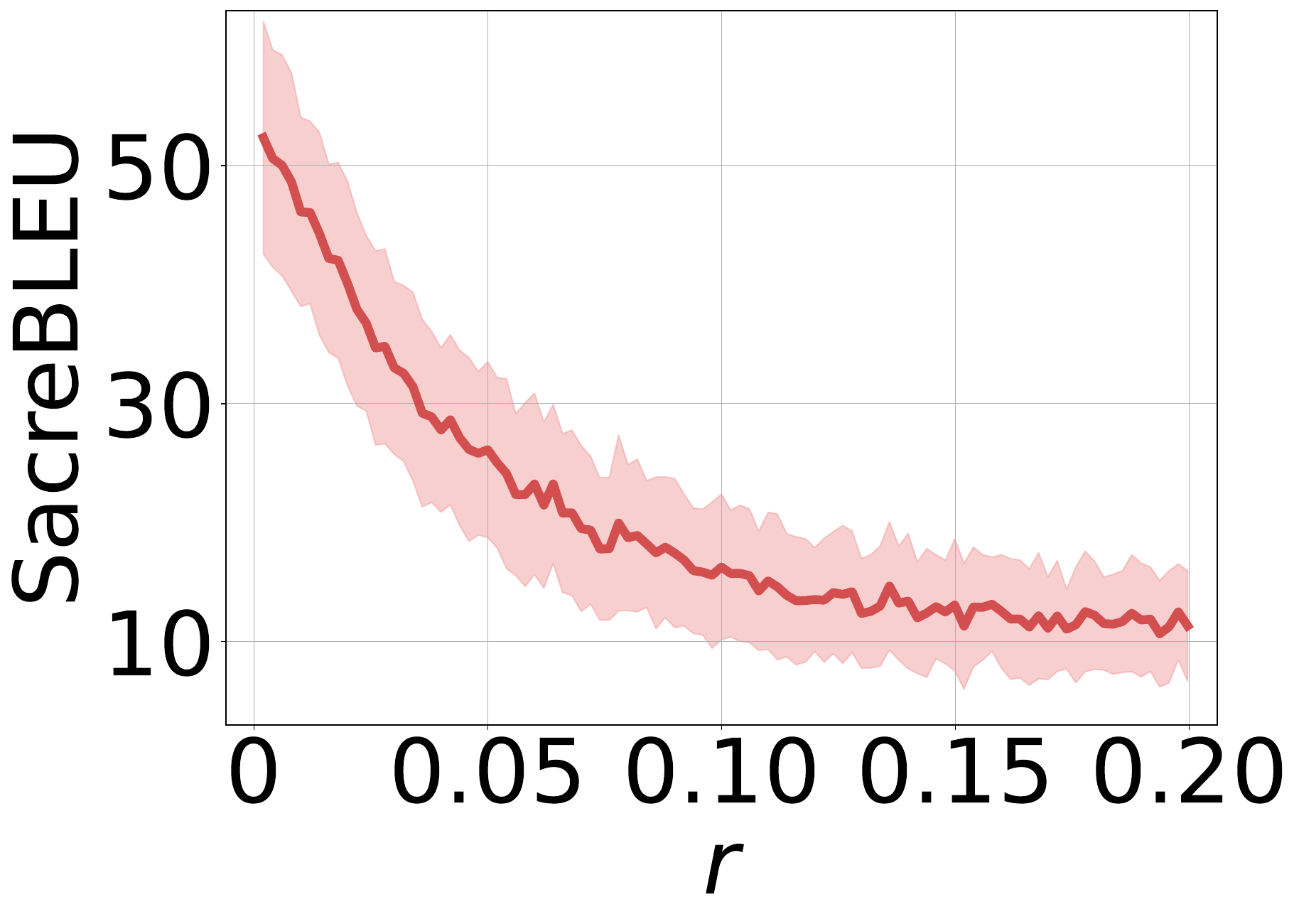}
        \caption{SacreBLEU-$r$}
        \label{fig:attack-r}
    \end{subcaptionblock}%
    \hfill%
    \begin{subcaptionblock}[t]{0.5\columnwidth}
        \centering
        \includegraphics[width=\textwidth]{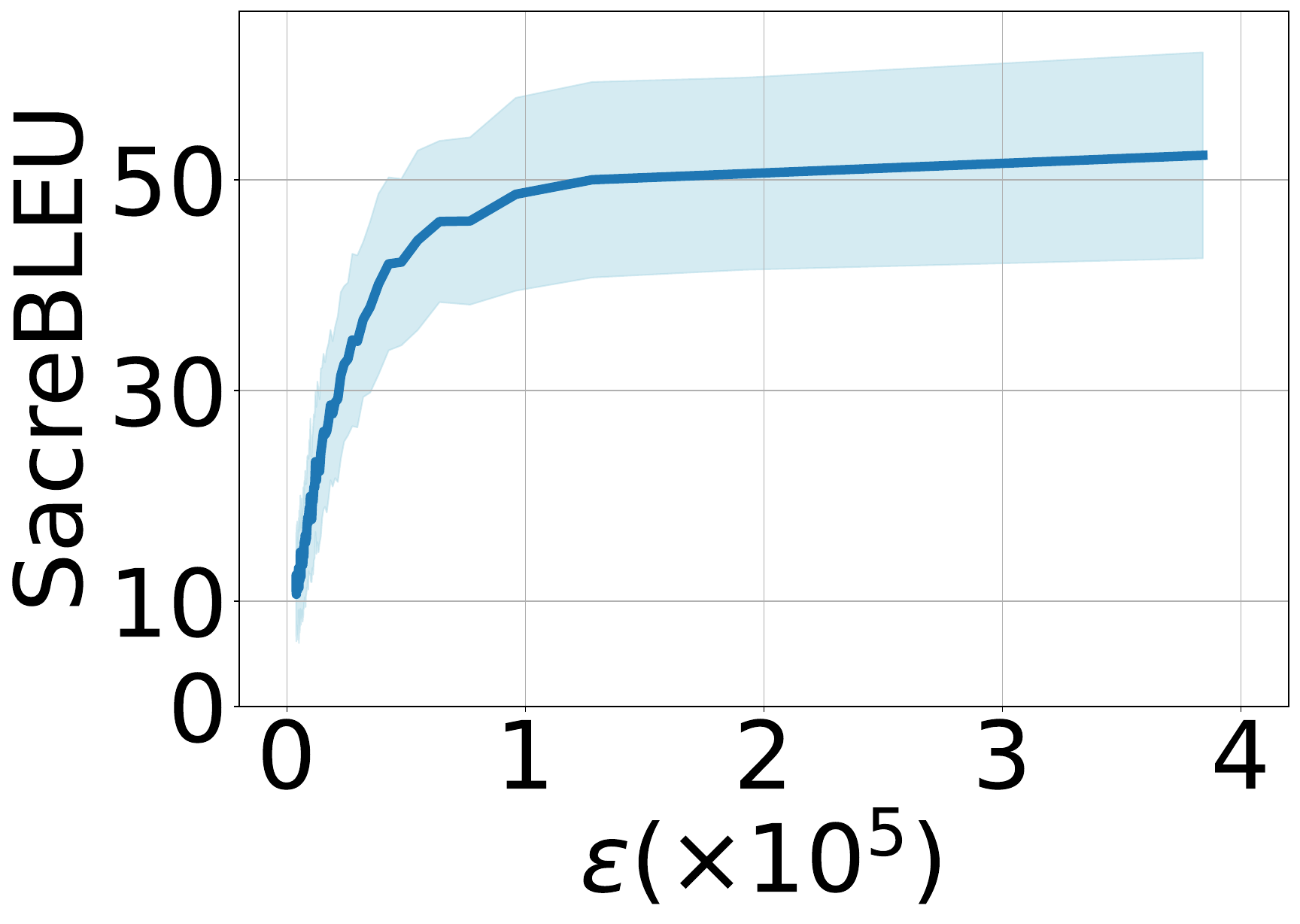}
        \caption{SacreBLEU-$\epsilon$}
        \label{fig:attack-eps}
    \end{subcaptionblock}%
    \caption{The attack performance corresponding to the perturbation $r$ and the privacy budget $\epsilon$.}
    \label{fig:attack}
\end{figure}

\begin{table*}[t]
    \caption{\ours achieves no loss in retrieval under various settings in our experiments.}
    \label{tab:acc}
    \centering
    \begin{tabular}{@{}cc@{}cc@{}cc@{}cc@{}}
        \toprule
                    & $N$                  & & $k$                & & $r$                        & & Embedding Model        \\
                      \cmidrule(lr){2-3}       \cmidrule(lr){4-5}     \cmidrule(lr){6-7}             \cmidrule(l){8-8}
                    & $10^4$/$10^5$/$10^6$ & & $5$/$10$/$15$/$20$ & & $0.03$/$0.05$/$0.07$/$0.1$ & & MiniLM/MPNet/T5/OpenAI-1/OpenAI-2 \\
        \midrule
        Recall & $100\%$              & & $100\%$            & & $100\%$                    & & $100\%$                \\
        \bottomrule
    \end{tabular}
\end{table*}

\subsection{Accuracy Study}
\label{subsec:accuracy}

To demonstrate the correctness of theoretical analysis for the calculation of $k'$, we conduct experiments under various settings:
different total numbers $N$ of documents, different numbers $k$ of top relevant documents, different sizes $r$ of the perturbation chosen by the user, and different embedding models.
We use recall to evaluate the proportion of top $k$ documents included in the results.

Throughout our experiment, we have not encountered a situation where any of the top $k$ documents are missing from the set of $k'$ documents.
As shown in \cref{tab:acc}, recall in all settings is $100\%$, indicating that all top $k$ documents are included in the set of $k'$ documents computed in module 1 and therefore can be correctly selected by module 2.

\begin{figure*}[t]
    \centering
    \begin{subcaptionblock}[t]{\textwidth}
        \centering
        \includegraphics[width=0.8\columnwidth]{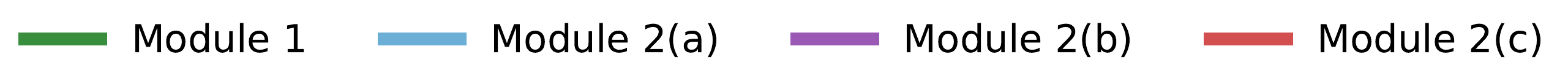}
    \end{subcaptionblock}\\
    \begin{subcaptionblock}[t]{0.33\textwidth}
        \centering
        \includegraphics[width=0.95\columnwidth]{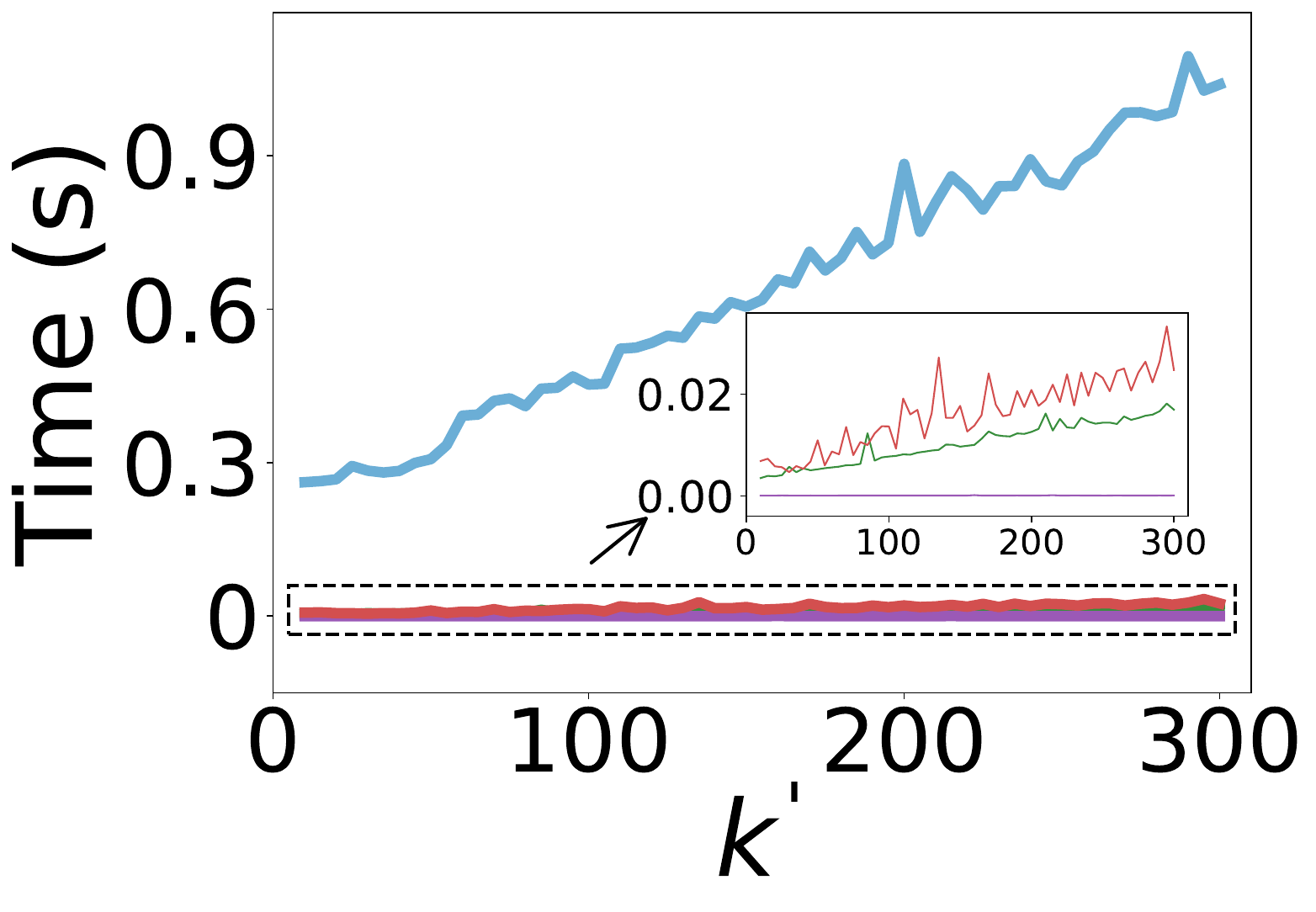}
        \caption{Computation cost.}
        \label{fig:computation}
    \end{subcaptionblock}%
    \hfill%
    \begin{subcaptionblock}[t]{0.33\textwidth}
        \centering
        \includegraphics[width=0.95\columnwidth]{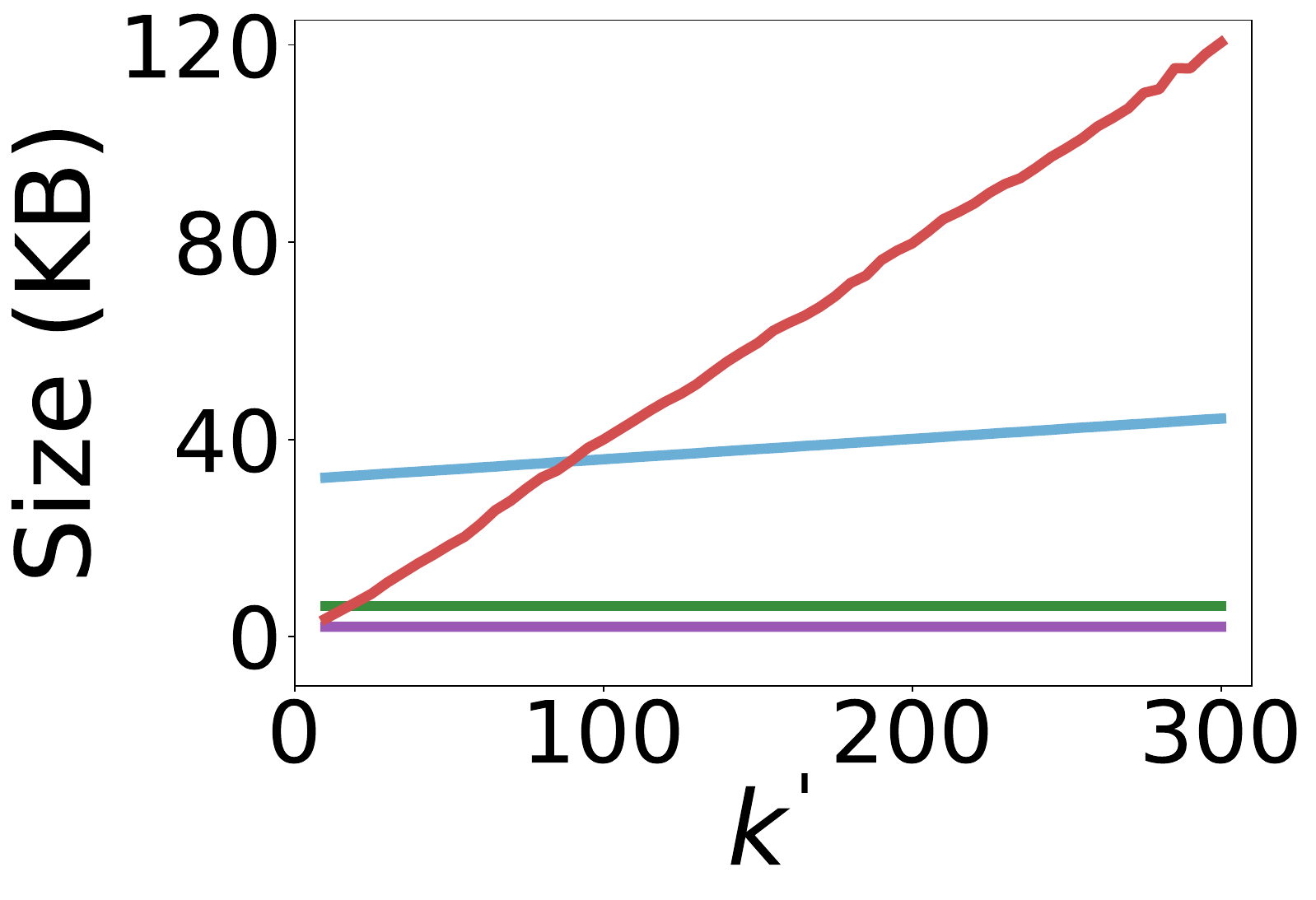}
        \caption{Communication cost.}
        \label{fig:communication}
    \end{subcaptionblock}%
    \hfill%
    \begin{subcaptionblock}[t]{0.33\textwidth}
        \centering
        \includegraphics[width=0.95\columnwidth]{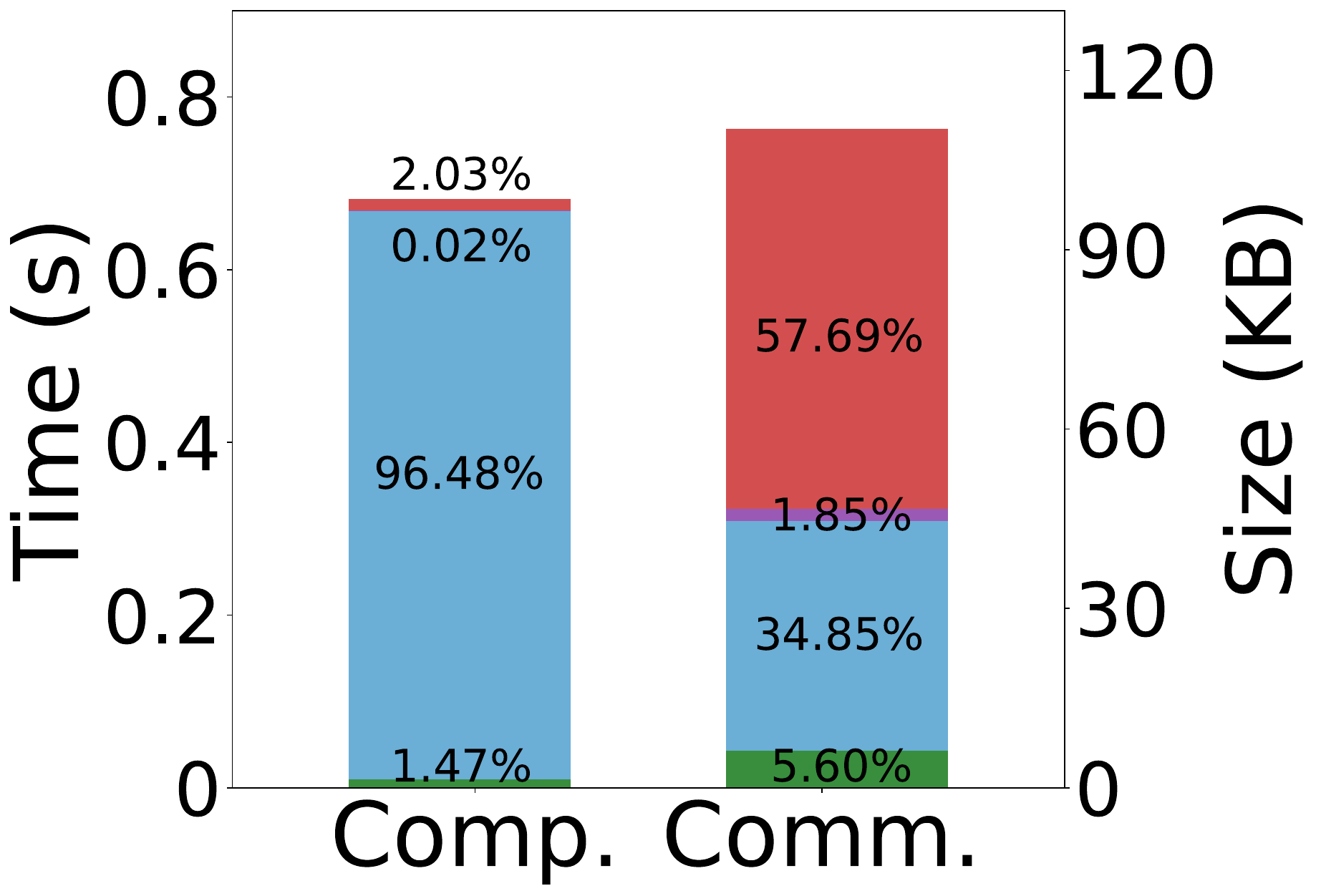}
        \caption{Cost breakdown ($k'=160$).}
        \label{fig:breakdown}
    \end{subcaptionblock}%
    \caption{Efficiency study of each module.}
    \label{fig:efficiency}
\end{figure*}

\subsection{Efficiency Study}
\label{subsec:efficiency}

For efficiency, we measure both computation and communication costs as a function of $k'$ to observe the impact of the perturbation on the costs.
The metrics are running time for computation cost and transmission size for communication cost.
We fix $k=5$ and choose T5 again as the embedding model.
We provide the results of each module in \cref{fig:efficiency}.
The linear results are consistent with the analysis in \cref{tab:special_cases}.
We highlight $k'=160,r=0.03$ in \cref{fig:breakdown} with the attack performance moderate at around $30$ shown in \cref{fig:attack}, and compare the results with baselines in \cref{tab:baseline}.

\textbf{Computation cost.}
From the results in \cref{fig:computation}, the most computationally intensive task occurs in module 2, which accounts for over $95\%$ of the total computation cost.
This substantial computational expense renders it impractical for scenarios involving a large number of documents when calculating cosine distances.
As indicated in \cref{tab:baseline}, the privacy-conscious service requires $2.72$ hours in total to process a single user request, which is considered unacceptable.

\textbf{Communication cost.}
From the results in \cref{fig:communication}, module 2(a) has a larger starting point, but the transmission size of module 2(c) soon surpasses module 2(a) as $k'$ increases.
Numerous basic parameters in PHE cause the former while the latter is due to the larger size of encrypted documents.
When $k'$ is relatively large, the transmission size becomes unacceptable.
As shown in \cref{tab:baseline}, the privacy-conscious service incurs a considerable transmission size ($1.43$GB) of $N$ documents.

\textbf{Direct and OT.}
At last, we compare the results using module 2(b) and module 2(c) in \cref{tab:baseline}.
We find little increase in computation cost but large increase in communication cost.
The cost breakdown illustrated in \cref{fig:breakdown} provides an intuitive distribution of costs.
OT does not bring much computation cost (from $0.02\%$ to $2.03\%$ compared to Direct), but the necessity to transfer $k'$ encrypted documents increases its transmission size from $1.85\%$ to $57.69\%$.

\section{Related Work}

RAG has garnered significant attention with the development of LLM.
Numerous researchers and companies \citep{DBLP_journals_corr_abs_2309_12871,voyageaiEmbeddings,DBLP_journals_corr_abs_2402_03216,OpenAI1,OpenAI2} are dedicated to training improved embedding models, aiming to group documents with similar semantic information more effectively.
Some \citep{llamaindexEvaluatingIdeal,langchainRecursivelySplit,towardsdatascienceAdvancedSmalltoBig} focus on chunking strategies to better split literature into documents at the appropriate locations and with the proper lengths.
Additionally, several works \citep{DBLP_conf_iclr_ZhouSHWS0SCBLC23,DBLP_journals_corr_abs_2309_11495,DBLP_journals_corr_abs_2305_14283} aim at query optimization to understand and transform queries for better retrieval results.

\textbf{These studies do not overlap with \ours.}
Instead, they can be directly applied to enhance the cloud RAG service or to transform the user query to improve the quality of RAG results.

\begin{table}[t]
    \caption{Efficiency comparison ($k'=160$).}
    \label{tab:baseline}
    \centering
    \begin{tabular}{@{}l@{}lr@{}lr@{}l@{}}
    \toprule
                           &                         & \multicolumn{2}{c}{Comp.}   & \multicolumn{2}{c@{}}{Comm.}  \\
    \midrule
    \multicolumn{2}{@{}l}{Privacy-ignorant Service}  & \texttt{3.15} & \texttt{ms} & \texttt{8.00}   & \texttt{KB} \\
    \multicolumn{2}{@{}l}{Privacy-conscious Service} & \texttt{2.72} & \texttt{hr} & \texttt{1.43}   & \texttt{GB} \\
    \\[-12pt]\hdashline\\[-12pt]
    \multirow{2}{*}{\ours} &(Direct)                 & \texttt{0.67} & \texttt{s}  & \texttt{46.66}  & \texttt{KB} \\
                           &(OT)                     & \texttt{0.68} & \texttt{s}  & \texttt{108.24} & \texttt{KB} \\
    \bottomrule
    \end{tabular}
\end{table}

\section{Conclusion}

In this paper, we are the first to address and formally define the privacy-preserving cloud RAG service problem.
We propose \ours as a solution regarding privacy, efficiency, and accuracy.
$(n,\epsilon)$-DistanceDP is introduced to characterize the privacy leakage of the user query.
The perturbation limits the search range, significantly saving computation and communication costs.
Theoretical analysis ensures the accuracy.
Experimental results also demonstrate the superiority of \ours in privacy, efficiency, and accuracy, compared to privacy-ignorant and privacy-conscious services.

\section*{Limitations}

\textbf{Limitation of PHE.}
PHE supports only the addition operation.
This restricts the variety of similarity distances \ours can calculate.
For example, FAISS offers to use Lp and Jaccard metrics, which may not be easy to use PHE.
Besides, RAG may also be combined with keyword searching for better retrieval results.
These require further investigation.

\textbf{Proprietary embedding model.}
Although open-source embedding models have already achieved great performance, the cloud may still consider using its own proprietary embedding model.
In this scenario, the user cannot calculate the query embedding locally and therefore cannot directly generate the perturbation for protection.
This remains a challenging problem.

\section*{Ethical Considerations}

\textbf{No disclosure risk.}
The privacy-preserving cloud RAG service is a new scenario proposed in this paper, which has not been formally used in practice.
\ours as the first solution to the potential privacy issues in this scenario can promote the development of this field with no disclosure risk.

\textbf{Open-sourced content in experiments.}
The open-sourced models and datasets used in our experiments are all downloaded from HuggingFace without modification.
We believe that using them appropriately according to their original purpose will not have a direct negative impact.

\textbf{Compliance with laws and regulations.}
\ours is proposed as a solution to potential privacy leakage in the privacy-preserving cloud RAG service, making it compliant with laws and regulations such as GDPR \citep{Voigt2017}.

% Bibliography entries for the entire Anthology, followed by custom entries
%\bibliography{anthology,custom}
% Custom bibliography entries only
\newpage
\bibliography{references.bib}

\appendix

\clearpage

\begin{algorithm*}[t]
    \caption{\ours: Module 1}
    \label{alg:module1}

    \DontPrintSemicolon

    \SetKwProg{mdone}{Module 1}{:}{}

    \mdone{}{
    \tcp{user side}
    generate a random perturbation $r\sim\mathrm{Gamma}(n,\frac{1}{\epsilon})$\;
    generate its direction vector $\mathbf{v}=\{v_1,\cdots,v_n\}$, where $v_i=\frac{t_i}{\sqrt{\sum_{j=1}^{n}{t_j}^2}},t_i\in\mathcal{N}(0,1),i\in[1,n]$\;
    compute the perturbed embedding $e_{k'}=e_k+r\mathbf{v}$\;
    determine $k'$ from \cref{thm:kpr_k}\;
    \tcp{cloud side}
    retrieve top $k'$ documents related to $e_{k'}$\;
    }
\end{algorithm*}

\begin{algorithm*}[t]
    \caption{\ours: Module 2}
    \label{alg:module2}

    \DontPrintSemicolon

    \SetKwProg{mdtwo}{Module 2}{:}{}

    \mdtwo{}{
    \tcp{user side}
    encrypt $\llbracket e_k\rrbracket$ using PHE\;
    \tcp{cloud side}
    \ForEach{$e\in$ \textnormal{document embeddings retrieved from \cref{alg:module1}}}{
    calculate cosine distance in encrypted form $\llbracket d_i\rrbracket=d_{\cos}(\llbracket e_k\rrbracket,e)$
    }
    \tcp{user side}
    decrypt cosine distances $d_i,i\in[1,k']$\;
    sort cosine distances to obtain the indices of top $k$ documents related to $e_k$\;
    \If{$\arctan\frac{\tan\alpha_k}{\sqrt{k}}\ge\frac{n}{\epsilon}$ \textnormal{(\cref{thm:index_privacy})}}{
        \tcp{cloud side}
        retrieve $k$ documents from the indices\;
    }
    \Else{
        retrieve $k$ documents from the $k$-out-of-$k'$ OT protocol\;
    }
    }
\end{algorithm*}

\section{More Details about \ours}

The detailed steps of \ours are shown in \cref{alg:module1,alg:module2}.

We did not emphasize the detail of the $k$-out-of-$k'$ OT protocol in module \hyperref[subsubsec:module2_c]{2(c)} in the main part, as it is not the primary contribution of this paper.
However, for privacy, efficiency, and accuracy, we provide the detail and analysis of the $k$-out-of-$k'$ OT protocol used in our experiments below.

\subsection{\texorpdfstring{\boldmath $k$-out-of-$k'$}{k-out-of-k'} Oblivious Transfer Protocol}
\label{subsec:ot}

We implement the $k$-out-of-$k'$ OT protocol based on \citet{DBLP_conf_latincrypt_ChouO15}.
Suppose the indices of target messages are $S=\{s_1,\cdots,s_k\}$, the protocol is as follows:

\begin{enumerate}[label=(\arabic*),leftmargin=*,noitemsep,topsep=0pt,wide=0pt]
    \item The cloud and the user share a hash function $\mathrm{Hash}$, a base number $g$ and a prime modulus $p$.
    \item The cloud selects a random number $a$, computes $A=g^a\bmod p$, and sends it to the user.
    \item The user computes $B_i=A^{c_i}\cdot g^{b_i}\bmod p,i\in[1,k']$, where $b_i$ are random numbers and $c_i=\left\{\begin{aligned}0,\  & i\in S\\1,\  & i\notin S\end{aligned}\right.$, and sends them to the cloud.
    \item The cloud constructs $k'$ secret keys $\mathrm{Key}_i=\mathrm{Hash}({B_i}^a\bmod p),i\in[1,k']$, uses them to encrypt messages $\llbracket m_i\rrbracket=\mathrm{Enc}(m_i,\mathrm{Key}_i),i\in[1,k']$, and sends these encrypted messages to the user.
    \item The user constructs $k$ secret keys $\mathrm{Key}_{s_j}=\mathrm{Hash}(A^{b_{s_j}}\bmod p),s_j\in S$, and can only decrypt target $k$ messages $m_{s_j}=\mathrm{Dec}(\llbracket m_{s_j}\rrbracket,\mathrm{Key}_{s_j})$.
\end{enumerate}

\textbf{Correctness.}
The objective is to ensure that keys used for encrypting and decrypting target messages are consistent between both sides, while keys corresponding to other messages remain inconsistent.
\begin{itemize}[label=$\triangleright$,leftmargin=*,noitemsep,topsep=0pt,wide=0pt]
\item For $i=s_j\in S,c_i=0$, the calculation of the key for $m_i$ on the cloud side is ${B_i}^a\equiv(A^{c_i}\cdot g^{b_i})^a\equiv g^{ab_i}\bmod p$.
Conversely, the calculation of the key for $m_{s_j}$ on the user side is $A^{b_{s_j}}\equiv g^{ab_{s_j}}\bmod p$.
This consistency in the calculation of the key on both sides enables the user to decrypt $m_{s_j}$.
\item For $i\notin S,c_i=1$, the calculation of the key for $m_i$ on the cloud side is ${B_i}^a\equiv(A^{c_i}\cdot g^{b_i})^a\equiv g^{a(a+b_i)}\bmod p$.
In contrast, the calculation of the key for $m_i$ on the user side is still $A^{b_i}\equiv g^{ab_i}\bmod p$, if the user insists on generating a key.
This inconsistency in the calculation of the key on both sides prevents the user from decrypting $m_{s_i}$.
\end{itemize}

\textbf{Security analysis.}
The cloud receives $B_i=g^{ac_i+b_i}\bmod p,i\in[1,k']$.
Since $b_i$ is a random number generated by the user, the cloud cannot derive whether $c_i=0$ or not by the given $B_i$.
Therefore, the cloud has no idea of which indices the user chooses.

\textbf{Communication analysis.}
There are $1.5$ rounds of communication.
First, the cloud sends a random number $A$, occupying $\beta$ units.
Second, the user sends $B_i,i\in[1,k']$ to the cloud, occupying $k'\beta$ units.
Third, the user receives encrypted messages (documents) from the cloud, occupying $k'\eta$ units.

\section{More Details in Experiments}

\subsection{Environment}

All of our experiments are conducted using PyTorch \citep{DBLP_conf_nips_PaszkeGMLBCKLGA19}, LangChain \citep{Chase_LangChain_2022} and HuggingFace \citep{DBLP_conf_emnlp_WolfDSCDMCRLFDS20} on an Ubuntu 22.04 server.
The server is equipped with two 28-core Intel(R) Xeon(R) Gold 5420+ processors and two Nvidia A40 48GB GPUs.

\subsection{Embedding Model Details}

\begin{table*}[t]
    \centering
    \caption{Patermeters of different embedding models in our experiments.}
    \label{tab:embedding_models}
    \begin{tabular}{@{}lccc@{}}
    \toprule
        Model                                                                          & License     & Model Size  & Dimension  \\
        \midrule
        all-MiniLM-L12-v2 \citep{huggingfaceSentencetransformersallMiniLML12v2Hugging} & Apache-2.0  & 33.4M       & 384 \\
        all-mpnet-base-v2 \citep{huggingfaceSentencetransformersallmpnetbasev2Hugging} & Apache-2.0  & 109M        & 768 \\
        gtr-t5-base \citep{huggingfaceSentencetransformersgtrt5baseHugging}            & Apache-2.0  & 110M        & 768 \\
        text-embedding-ada-002 \citep{OpenAI1}                                         & Proprietary & Proprietary & 1536 \\
        text-embedding-3-large \citep{OpenAI2}                                         & Proprietary & Proprietary & 3072 \\
         \bottomrule
    \end{tabular}
\end{table*}

We use three open-sourced embedding models and two OpenAI proprietary embedding models in our experiments.
The detailed parameters are shown in \cref{tab:embedding_models}.
OpenAI has not explicitly stated the license and model size for their embedding models.

\subsection{Dataset Details}

We use MS MARCO \citep{DBLP_conf_nips_NguyenRSGTMD16} v1.1 as the dataset in our experiments.
To be specific, We extract lists of \texttt{passage\_text} from the column \texttt{passages}.
We randomly choose $10^4$, $10^5$, and $10^6$ \texttt{passage\_text} to construct $15$ vector databases in total with five different embedding models.
For simplicity, we make no modifications to these passages such as splitting.

\subsection{Default Setting}

Unless otherwise specified, our experiments used the following hyperparameter settings:
$N=10^5$, $k=5$, $k'=160$, and T5 as the embedding model.
The experimental results shown in all figures and tables are the averages of 50 independent experiments.

\subsection{Extra Experimental Results}

\subsubsection{Relationships among Hyperparameters}

\textbf{\boldmath $k/N$-$\alpha_k$.}
We first plot the equation of \cref{lem:k_alpha_k}.
From \cref{fig:k-N-alpha_k}, the result of $k/N$ increases sharply as $\alpha_k$ approaches $90\degree$.
Additionally, when $n$ grows larger, the increase is even steeper.
This phenomenon is characteristic of high-dimensional space, where random vectors on the surface of the unit $n$-sphere tend to be almost perpendicular.
Consequently, a relatively small change in $\alpha_k$ results in a significant change in $k/N$, meaning that the perturbation greatly impacts $k'$, highlighting the importance of selecting a proper privacy budget.

\textbf{\boldmath $\epsilon$-$k'$.}
We discuss in \cref{subsec:module1} that in practice, apart from initially setting the privacy budget, we can also choose $k'$ first and then compute the corresponding privacy budget.
From \cref{fig:epsilon-kpr}, we observe that when $k'<50$, the change in $\epsilon$ is relatively large, which corresponds to a small perturbation and high attack performance according to \cref{fig:attack}.
The user should avoid considering $k'$ as well as the corresponding privacy budget in this range, since the protection is too weak, as shown in \cref{fig:attack}.
To avoid excessive computation and communication costs, an appropriate choice of $k'$ would be within the range of $[100,200]$.
Another observation is that a larger $\epsilon$ is required to preserve the same value of $k'$ for a larger value of $k$.
This can be explained by the fact that, as $k$ increases, the number of possible embeddings with the same top $k$ documents also increases.
Therefore, the same value of $k'$ implies a looser privacy requirement, which is reflected by a larger privacy budget $\epsilon$.

\begin{figure}[t]
    \centering
    \begin{subcaptionblock}{0.48\columnwidth}
        \centering
        \includegraphics[height=0.13\textheight]{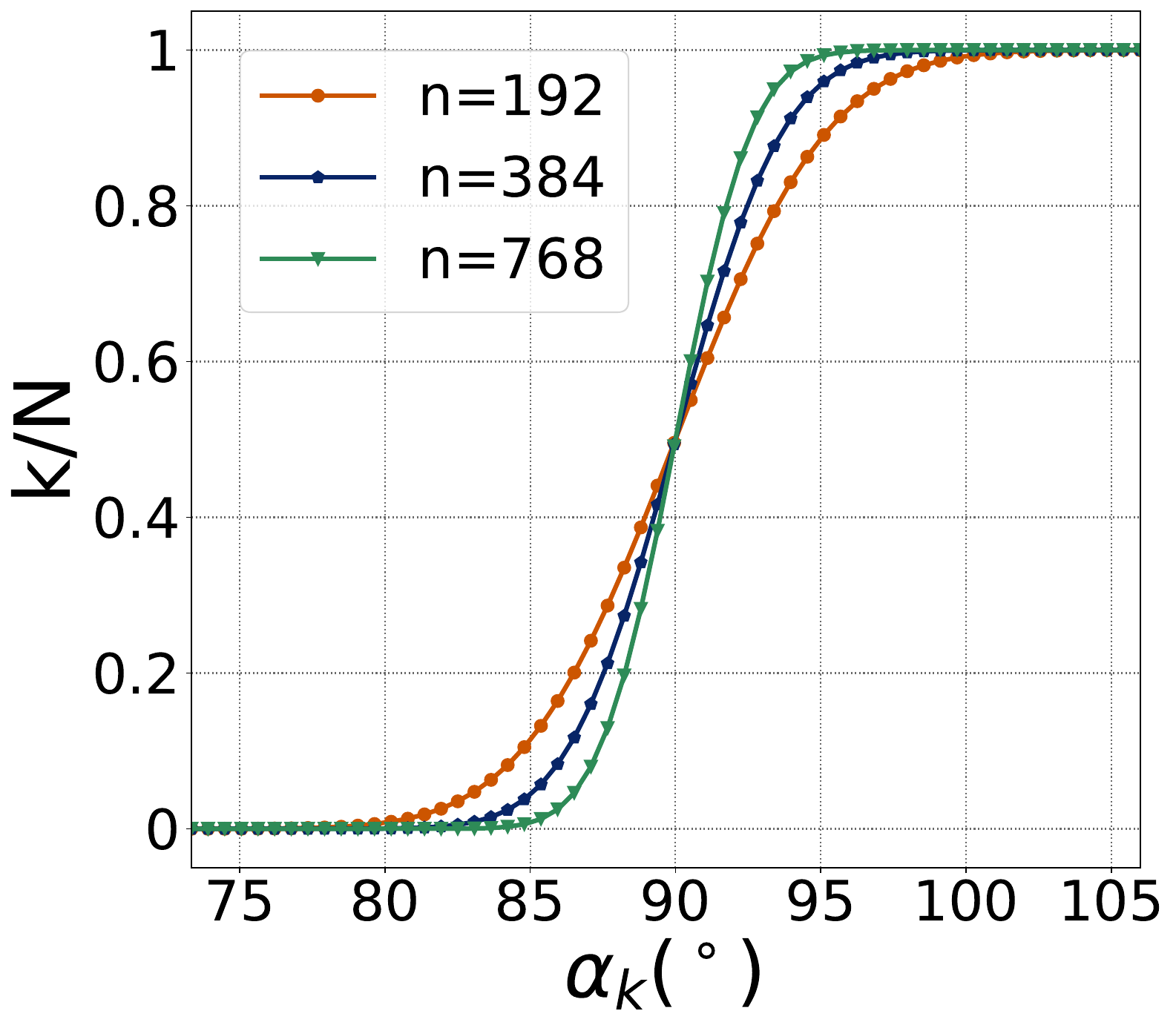}
        \caption{$k/N$-$\alpha_k$}
        \label{fig:k-N-alpha_k}
    \end{subcaptionblock}%
    \hfill%
    \begin{subcaptionblock}{0.48\columnwidth}
        \centering
        \includegraphics[height=0.13\textheight]{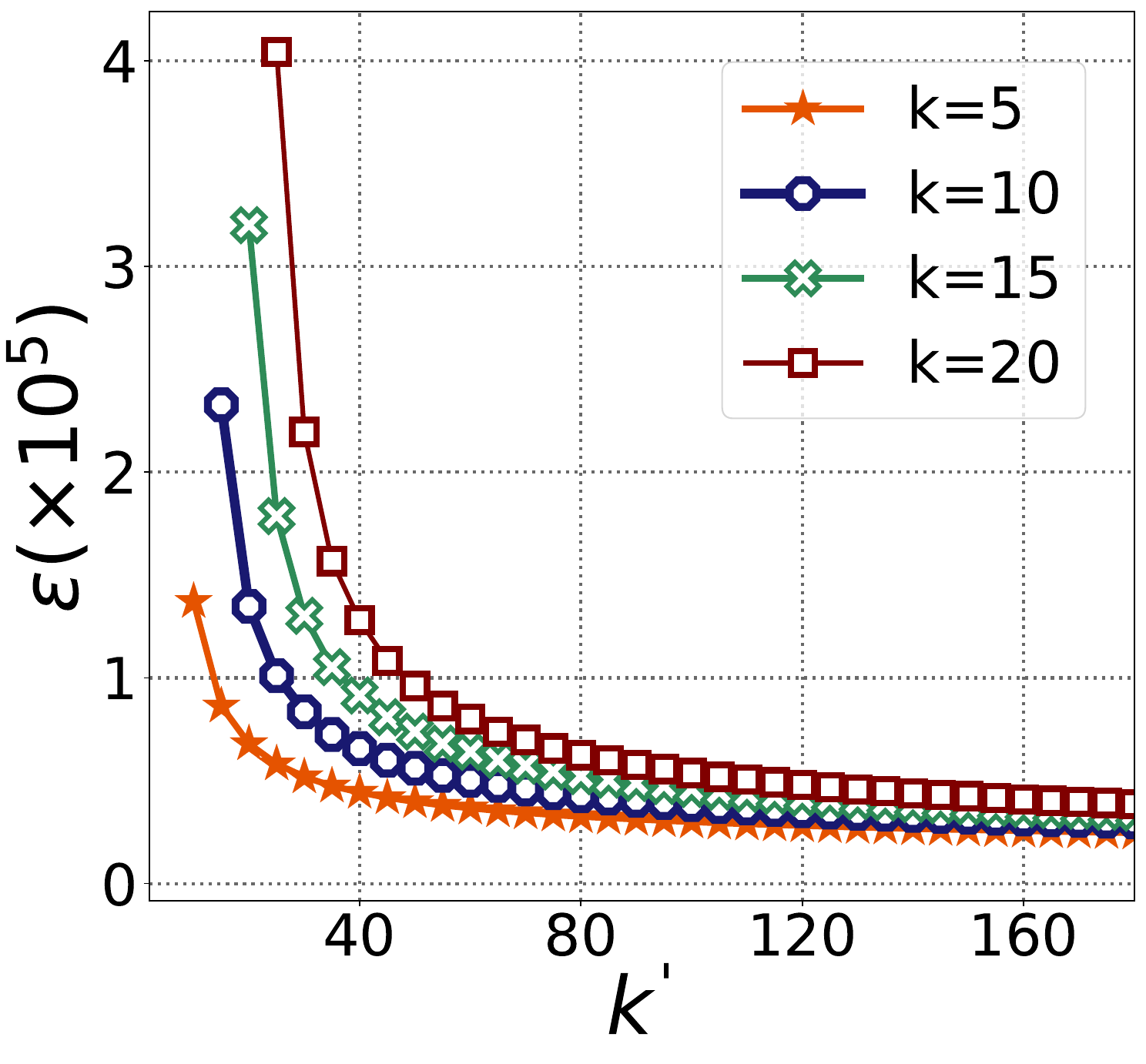}
        \caption{$\epsilon$-$k'$}
        \label{fig:epsilon-kpr}
    \end{subcaptionblock}
    \caption{Relationships among hyperparameters.}
    \label{fig:relationship}
\end{figure}

\subsubsection{2D Simulations}

The following two 2D simulations are based on some artificial data for ease of understanding.

\textbf{\boldmath Visualization of Top $k$.}
To provide a clear view of how \ours works, we give a visualization of the relationship between the embeddings of the top $k$ documents and the selected $k'$ documents in $2$-dimensional space.
The plot in \cref{fig:visualization} clearly shows that the top $k$ documents are included in the set of $k'$ documents, demonstrating the correctness of \ours.

\textbf{Rare exceptions.}
Although we have not experienced any loss in retrieving documents, we acknowledge that in some rare exceptions, there might be a chance of \ours failing to preserve the top $k$ documents.
We provide one such exception here.
As illustrated in \cref{fig:exception}, the top $k$ documents related to the query embedding are all located on the left side of the query embedding at the same angle $\alpha_k$.
The perturbed embedding is positioned on the right side of the query embedding at angle $\Delta\alpha_k$.
If there are $k'$ documents located exactly within the range of angles $[\alpha_k,\alpha_k+2\Delta\alpha_k]$, then the top $k'$ documents related to $e_{k'}$ would not include the top $k$ documents related to $e_k$.
However, exceptions like this only occur when the distribution of document embeddings is extremely non-uniform, a rare scenario in high-dimensional space.

\begin{figure}[t]
    \centering
    \begin{subcaptionblock}{0.48\columnwidth}
        \centering
        \includegraphics[height=0.13\textheight]{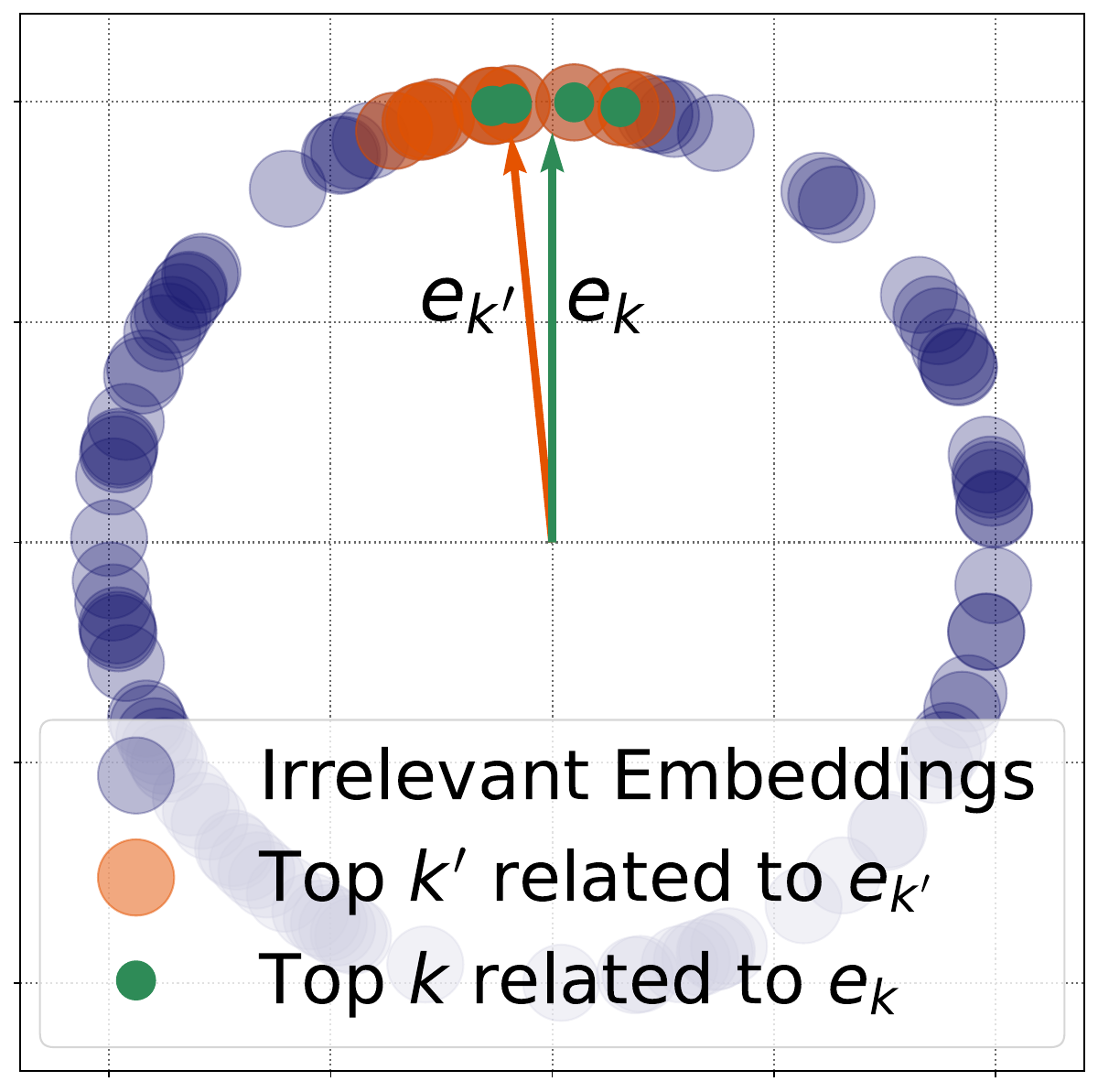}
        \caption{Visualization of top $k$.}
        \label{fig:visualization}
    \end{subcaptionblock}%
    \hfill%
    \begin{subcaptionblock}{0.48\columnwidth}
        \centering
        \includegraphics[height=0.13\textheight]{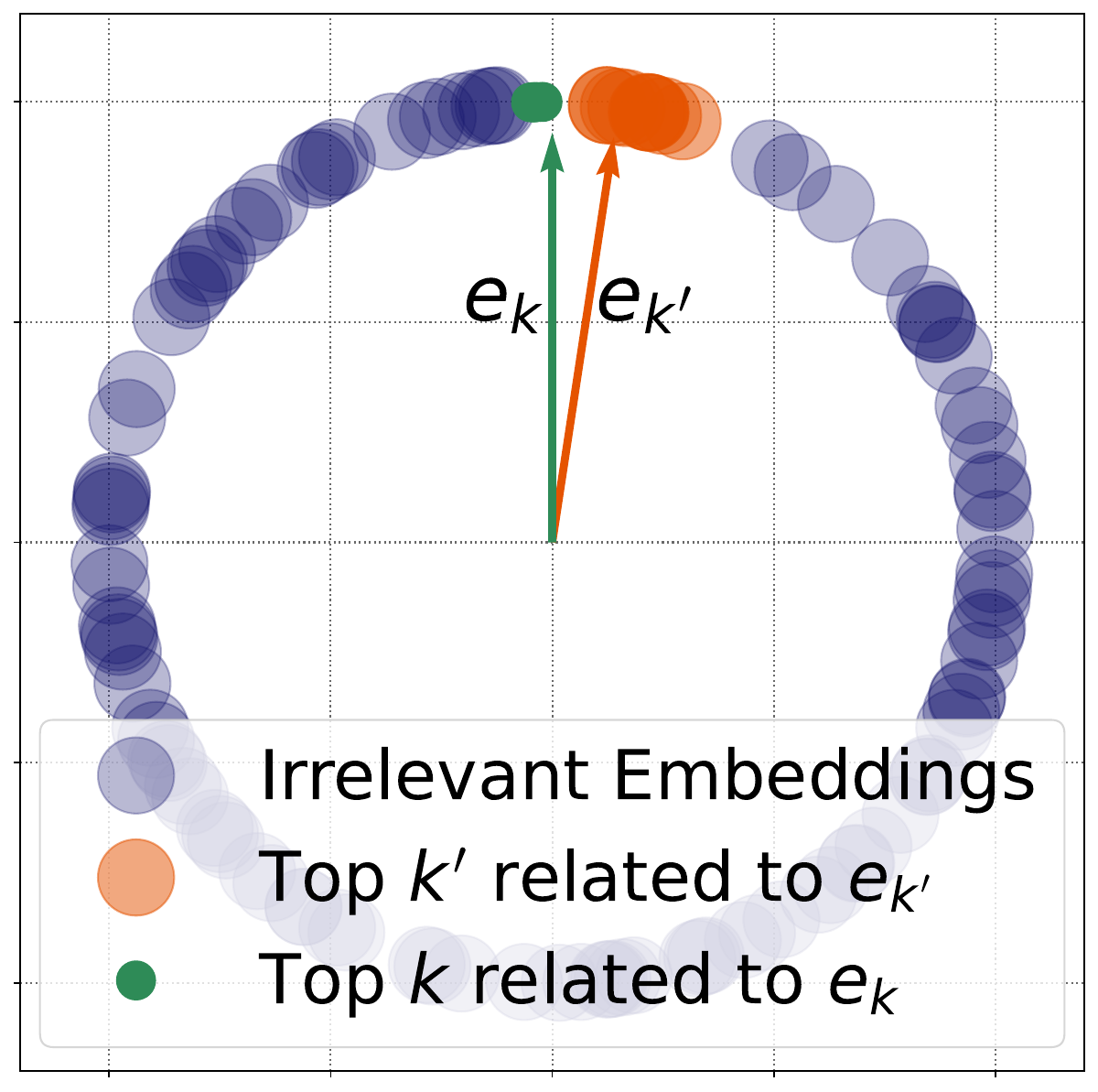}
        \caption{One rare exception.}
        \label{fig:exception}
    \end{subcaptionblock}%
    \caption{Illustrations in $2$-dimensional space.}
    \label{fig:acc}
\end{figure}

\section{Proofs}

\begin{lemma}[Repeated from \cref{lem:k_alpha_k}]
    \label{lem:k_alpha_k_repeated}
    Assume that there are $N$ embeddings uniformly distributed on the surface of the $n$-dimensional unit sphere.
    Let $\alpha_k$ be the polar angle of the surface area formed by top $k$ embeddings related to any given embedding.
    Then, $k$ and $\alpha_k$ satisfy the following relationship:
    \[k=N\cdot\frac{\Omega_{n-1}(\pi)}{\Omega_{n}(\pi)}\cdot\int_{0}^{\alpha_{k}}\sin^{n-2}\theta\,\mathrm{d}\theta\]
    where $\Omega_{n}(\pi)=\frac{2\pi^\frac{n}{2}}{\Gamma(\frac{n}{2})}$ represents the surface area of the unit $n$-sphere.
\end{lemma}

\begin{figure}[ht]
    \centering
    \includegraphics[width=0.9\columnwidth]{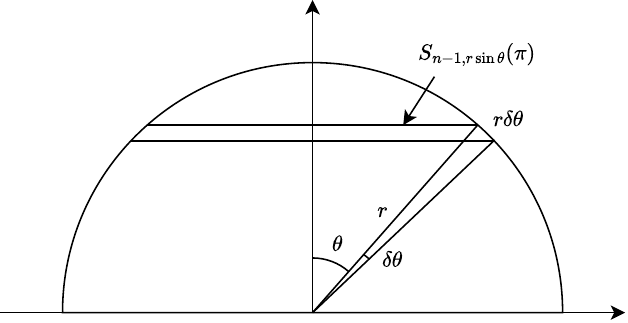}
    \caption{Illustration of the proof of \cref{lem:k_alpha_k_repeated}.}
    \label{fig:lem1proof}
\end{figure}

\begin{proof}
    Define $S_{n,r}(\alpha)=\Omega_n(\alpha)r^{n-1}$ as the surface area of the spherical sector with a polar angle $\alpha\in[0,\pi]$ in the $n$-sphere with radius $r$, where $\Omega_n(\alpha)$ represents the corresponding surface area in the unit $n$-sphere.
    Then, referring to \cref{fig:lem1proof}, we have
    \begin{align*}
        S_{n,r}(\alpha) & =\int_{0}^{\alpha}S_{n-1,r\sin\theta}(\pi)r\,\mathrm{d}\theta                \\
                        & =\int_{0}^{\alpha}\Omega_{n-1}(\pi)[r\sin\theta]^{n-2}r\,\mathrm{d}\theta    \\
                        & =r^{n-1}\int_{0}^{\alpha}\Omega_{n-1}(\pi)\sin^{n-2}\theta\,\mathrm{d}\theta
    \end{align*}

    Comparing to the definition of $S_{n,r}(\alpha)$, it is straightforward to derive that
    \[\Omega_n(\alpha)=\int_{0}^{\alpha}\Omega_{n-1}(\pi)\sin^{n-2}\theta\,\mathrm{d}\theta\]

    Assuming the embeddings are uniformly distributed on the surface,
    \[\frac{N}{\Omega_n(\pi)}=\frac{k}{\Omega_n(\alpha_k)}\]

    Therefore,
    \[k=N\cdot\frac{\Omega_n(\alpha_k)}{\Omega_n(\pi)}=N\cdot\frac{\Omega_{n-1}(\pi)}{\Omega_{n}(\pi)}\cdot\int_{0}^{\alpha_{k}}\sin^{n-2}\theta\,\mathrm{d}\theta\]
\end{proof}

\begin{theorem}[Repeated from \cref{thm:kpr_k}]
    Under the conditions specified in \cref{lem:k_alpha_k_repeated},
    given two embeddings $e_k$ and $e_{k'}$ with the perturbed angle $\Delta\alpha_k$, to ensure that top $k'$ embeddings related to $e_{k'}$ include top $k$ embeddings related to $e_k$, $k'$ and $k$ satisfy the following relationship:
    \[\Delta k=k'-k=N\cdot\frac{\Omega_{n-1}(\pi)}{\Omega_{n}(\pi)}\cdot\int_{\alpha_k}^{\alpha_{k'}}\sin^{n-2}\theta\,\mathrm{d}\theta\]
    where $\alpha_{k'}=\alpha_{k}+\Delta\alpha_k$.
\end{theorem}

\begin{proof}
    From \cref{fig:thm1proof_special}, we observe that $\alpha_{k'}=\alpha_{k}+\Delta\alpha_k$.
    This property can be readily extended to $n$-dimensional space.
    And from \cref{lem:k_alpha_k_repeated},
    \begin{align*}
        k&=N\cdot\frac{\Omega_{n-1}(\pi)}{\Omega_{n}(\pi)}\cdot\int_{0}^{\alpha_{k}}\sin^{n-2}\theta\,\mathrm{d}\theta  \\
        k'&=N\cdot\frac{\Omega_{n-1}(\pi)}{\Omega_{n}(\pi)}\cdot\int_{0}^{\alpha_{k'}}\sin^{n-2}\theta\,\mathrm{d}\theta
    \end{align*}

    Thus,
    \[\Delta k=k'-k=N\cdot\frac{\Omega_{n-1}(\pi)}{\Omega_{n}(\pi)}\cdot\int_{\alpha_k}^{\alpha_{k'}}\sin^{n-2}\theta\,\mathrm{d}\theta\]
\end{proof}

\begin{theorem}[Repeated from \cref{thm:distance}]
    Given two normalized embeddings $e_a$ and $e_b$ of the same dimension, L2 distance and cosine distance have the following relationship:
    \[d_{l2}(e_a,e_b)=\sqrt{2d_{\cos}(e_a,e_b)}\]
\end{theorem}

\begin{proof}
    From \cref{def:cos_dis},
    \begin{align*}
        d_{l2}^2(e_a,e_b)&=\|e_a-e_b\|^2=\sum_{i=1}^{n}({e_a}_i-{e_b}_i)^2\\
        &=\sum_{i=1}^{n}{e_a}_i^2+\sum_{i=1}^{n}{e_b}_i^2-\sum_{i=1}^{n}2{e_a}_i{e_b}_i\\
        &=\|e_a\|^2+\|e_b\|^2-2\sum_{i=1}^{n}{e_a}_i{e_b}_i\\
        &=1+1-2\langle e_a,e_b\rangle\\
        &=2d_{\cos}(e_a,e_b)
    \end{align*}
\end{proof}

\begin{lemma}
    \label{lem:mean}
    $k$ points $p_1,\cdots,p_k$ are extracted from the uniform distribution on the surface of an $n$-dimensional sphere with radius $r$.
    Denote the mean of these points as $\overline{p}$.
    L2 distance $d$ between $\overline{p}$ and the center of the sphere has the expected value
    \[\mathbb{E}[d]=\frac{r}{\sqrt{k}}\]
\end{lemma}

\begin{proof}
    We place the center of the sphere at the origin.
    Since $p_i=\{x_{i1},\cdots,x_{in}\},i\in[1,k]$ is extracted from the uniform distribution on the surface of an $n$-dimensional sphere with radius $r$, the coordinates can be contructed by two steps: generating $y_{ij}\sim\mathcal{N}(0,1)$ and $x_{ij}=r\cdot\frac{y_{ij}}{\sqrt{\sum_{m=1}^{n}y_{im}^2}},j\in[1,n]$.

    Due to symmetry, $\mathbb{E}[x_{ij}]=0$,
    \begin{align*}
        &\mathbb{E}\left[\frac{y_{i1}^2}{\sum_{m=1}^{n}y_{im}^2}\right]=\cdots=\mathbb{E}\left[\frac{y_{in}^2}{\sum_{m=1}^{n}y_{im}^2}\right]\\
        =&\frac{1}{n}\cdot\sum_{j=1}^{n}\mathbb{E}\left[\frac{y_{ij}^2}{\sum_{m=1}^{n}y_{im}^2}\right]=\frac{1}{n}\cdot\mathbb{E}\left[\frac{\sum_{j=1}^{n}y_{ij}^2}{\sum_{m=1}^{n}y_{im}^2}\right]=\frac{1}{n}
    \end{align*}
    and
    \begin{align*}
        \mathrm{Var}(x_{ij})&=\mathbb{E}[x_{ij}^2]-(\mathbb{E}[x_{ij}])^2\\
        &=\mathbb{E}\left[r^2\cdot\frac{y_{ij}^2}{\sum_{m=1}^{n}y_{im}^2}\right]=\frac{r^2}{n}
    \end{align*}

    By the central limit theorem, each coordinate component $\overline{x_j}$ of $\overline{p}$ satisfies
    \[\overline{x_j}=\frac{1}{k}\cdot\sum_{i=1}^{k}x_{ij}\sim\mathcal{N}\left(0,\frac{r^2}{kn}\right),\frac{\sqrt{kn}}{r}\cdot\overline{x_j}\sim\mathcal{N}(0,1)\]

    Notice that
    \[\frac{kn}{r^2}\cdot d^2=\frac{kn}{r^2}\cdot\sum_{j=1}^{n}\overline{x_j}^2=\sum_{j=1}^{n}\left(\frac{\sqrt{kn}}{r}\cdot\overline{x_j}\right)^2\sim\chi^2(n)\]

    Thus, $\frac{kn}{r^2}\cdot\mathbb{E}\left[d^2\right]=n$ and $\mathbb{E}[d]=\frac{r}{\sqrt{k}}$.
\end{proof}

\begin{theorem}[Repeated from \cref{thm:index_privacy}]
    \label{thm:index_privacy_repeated}
    Given a target query embedding $e_k$ and the mean embedding $\overline{e}$ of top $k$ relevant document embeddings, the mean angle $\omega$ between $e_k$ and $\overline{e}$ satisfies
    \begin{equation*}
        \tan\omega=\frac{\tan\alpha_k}{\sqrt{k}}
    \end{equation*}
    where $\alpha_k$ is calculated from \cref{lem:k_alpha_k_repeated}.
\end{theorem}

\begin{figure}[ht]
    \centering
    \includegraphics[width=0.9\columnwidth]{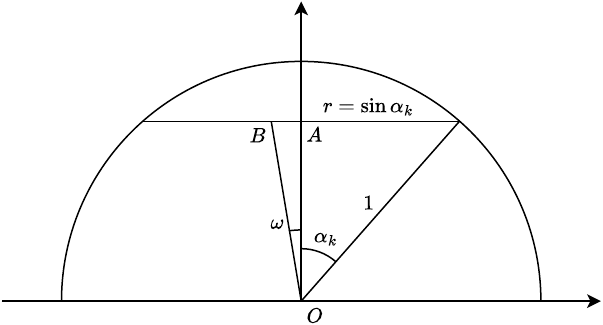}
    \caption{Illustration of the proof of \cref{thm:index_privacy_repeated}.}
    \label{fig:thm6proof}
\end{figure}

\begin{proof}
    \cref{lem:k_alpha_k_repeated} tells us that top $k$ embeddings are within the polar angle $\alpha_k$.
    When $n\gg1$, we approximately believe that the angle they make with $e_k$ is exactly $\alpha_k$, which means $k$ embeddings are uniformly distributed on the surface of an ($n-1$)-dimensional sphere with radius $\sin\alpha_k$.

    Applying \cref{lem:mean} and referring to \cref{fig:thm6proof}, $\mathbb{E}[AB]=\frac{\sin\alpha_k}{\sqrt{k}}$.
    Since $OA=\cos\alpha_k$,
    \[\tan\omega=\frac{\mathbb{E}[AB]}{OA}=\frac{\sin\alpha_k}{\cos\alpha_k\sqrt{k}}=\frac{\tan\alpha_k}{\sqrt{k}}\]
\end{proof}

\end{document}